\documentclass[a4paper,UKenglish,cleveref, autoref, thm-restate]{lipics-v2021}
\usepackage{xspace}

\hideLIPIcs  

\title{5-Approximation for $\hh$-Treewidth Essentially as Fast as $\hh$-Deletion Parameterized by Solution Size}

\titlerunning{5-Approximation for $\hh$-Treewidth} 

\author{Bart M.\,P. Jansen 
}{Eindhoven University of Technology, The Netherlands}{b.m.p.jansen@tue.nl}{https://orcid.org/0000-0001-8204-1268}{}

 \author{Jari J.\,H. {de Kroon}}{Eindhoven University of Technology, The Netherlands}{j.j.h.d.kroon@tue.nl}{https://orcid.org/0000-0003-3328-9712}{}

 \author{Micha\l{} W\l{}odarczyk}{University of Warsaw, Poland}{m.wlodarczyk@tue.nl}{https://orcid.org/0000-0003-0968-8414}{}

\authorrunning{B.M.P. Jansen, J.J.H. de Kroon, and M. W\l{}odarczyk}

\funding{This project has received funding from the European Research Council (ERC) under the European Union's Horizon 2020 research and innovation programme (grant agreement No 803421, ReduceSearch).\\
\includegraphics[width=4cm]{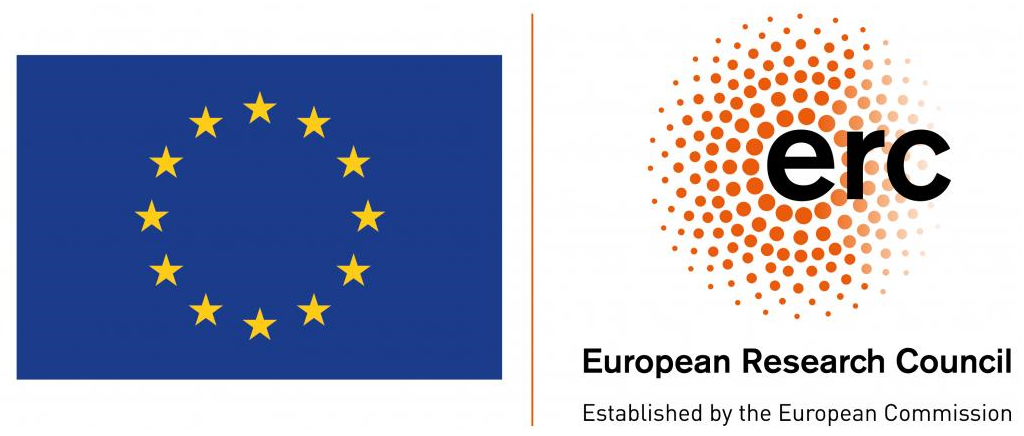}} 

\Copyright{Bart M.\,P. Jansen, Jari J. H. de Kroon, and Micha\l{} W\l{}odarczyk}



\category{} 

\relatedversion{} 

\nolinenumbers 

\EventEditors{John Q. Open and Joan R. Access}
\EventNoEds{2}
\EventLongTitle{42nd Conference on Very Important Topics (CVIT 2016)}
\EventShortTitle{CVIT 2016}
\EventAcronym{CVIT}
\EventYear{2016}
\EventDate{December 24--27, 2016}
\EventLocation{Little Whinging, United Kingdom}
\EventLogo{}
\SeriesVolume{42}
\ArticleNo{23}

\newcommand{\defproblem}[3]{
	\vspace{2mm}
	\noindent\fbox{
		\begin{minipage}{0.96\linewidth}
			\begin{tabular*}{\linewidth}{@{\extracolsep{\fill}}lr} \textsc{#1} &  \\ \end{tabular*}
			{\bf{Input:}} #2 \\
			{\bf{Task:}} #3
		\end{minipage}
	}
	\vspace{2mm}
}

\newcommand{\tw}{\mathrm{\textbf{tw}}}

\newcommand{\ed}{\mathrm{\textbf{ed}}}

\newcommand{\Oh}{\mathcal{O}}
\newcommand{\hh}{\ensuremath{\mathcal{H}}}

\newcommand{\sm}{\setminus}

\newcommand{\hhtw}[1][\hh]{\tw_{#1}}
\newcommand{\hhdepth}[1][\hh]{\ed_{#1}}
\newcommand{\hhtwfull}[1][\hh]{{#1}-treewidth}
\newcommand{\hhdepthfull}[1][\hh]{{#1}-elimination distance}

\newcommand{\hhdel}{\textsc{$\hh$-deletion}\xspace}

\newcommand{\hsepk}{$(\hh,k)$-separation\xspace}

\ifdefined\DEBUG{}
\newcommand{\mic}[1]{{\color{blue}{#1}}}

\def\rem#1{{\marginpar{\raggedright\scriptsize #1}}}
\newcommand{\micr}[1]{\rem{\textcolor{blue}{\(\bullet \) #1}}}

\newcommand{\bmp}[1]{{\color{purple}{#1}}}
\newcommand{\bmpr}[1]{\rem{\textcolor{purple}{\(\bullet \) #1}}}

\newcommand{\jjh}[1]{{\color{orange}{#1}}}
\newcommand{\jjhr}[1]{\rem{\textcolor{orange}{\(\bullet \) #1}}}

\else
\newcommand{\mic}[1]{#1}
\newcommand{\bmp}[1]{#1}
\newcommand{\jjh}[1]{#1}
\newcommand{\micr}[1]{ }
\newcommand{\bmpr}[1]{ }
\newcommand{\jjhr}[1]{ }
\fi

\begin{document}

\maketitle

\begin{abstract}
The notion of $\hh$-treewidth, where~$\hh$ is a hereditary graph class, was recently introduced as a generalization of the treewidth of an undirected graph. Roughly speaking, a graph of $\hh$-treewidth at most~$k$ can be decomposed into (arbitrarily large) $\hh$-subgraphs which interact only through vertex sets of size~$\Oh(k)$ which can be organized in a tree-like fashion. $\hh$-treewidth can be used as a hybrid parameterization to develop fixed-parameter tractable algorithms for \hhdel problems, which ask to find a minimum vertex set whose removal from a given graph~$G$ turns it into a member of~$\hh$. The bottleneck in the current parameterized algorithms lies in the computation of suitable tree $\hh$-decompositions. 

We present FPT approximation algorithms to compute tree $\hh$-decompositions for hereditary and union-closed graph classes~$\hh$. Given a graph of $\mathcal{H}$-treewidth~$k$, we can compute a 5-approximate tree $\mathcal{H}$-decomposition in time~$f(\Oh(k)) \cdot n^{\Oh(1)}$ whenever \textsc{$\mathcal{H}$-deletion} parameterized by solution size can be solved in time~$f(k) \cdot n^{\Oh(1)}$ for some function~$f(k) \geq 2^k$.
The current-best algorithms either achieve an approximation factor of $k^{\Oh(1)}$ or construct optimal decompositions while suffering from non-uniformity with unknown parameter dependence. Using these decompositions, we obtain algorithms solving \textsc{Odd Cycle Transversal} in time~$2^{\Oh(k)} \cdot n^{\Oh(1)}$ parameterized by $\mathsf{bipartite}$-treewidth and \textsc{Vertex Planarization} in time~$2^{\Oh(k \log k)} \cdot n^{\Oh(1)}$ parameterized by $\mathsf{planar}$-treewidth, showing that these can be as fast as the solution-size parameterizations and giving the first ETH-tight algorithms for parameterizations by hybrid width measures.
\end{abstract}

\clearpage

\section{Introduction} \label{sec:intro}

\subparagraph{Background and motivation.} 

Treewidth (see~\cite{Bodlaender98,Diestel12} \cite[\S 7]{CyganFKLMPPS15}) is a width measure for graphs that is ubiquitous in algorithmic graph theory. It features prominently in the Graph Minors series~\cite{RobertsonS90a} and frequently pops up unexpectedly~\cite{Marx20} in the parameterized complexity~\cite{CyganFKLMPPS15,DowneyF13,DBLP:series/txtcs/FlumG06} analysis of NP-hard graph problems on undirected graphs. The notion of treewidth captures how tree-like a graph is in a certain sense; it is defined as the width of an optimal tree decomposition for the graph. Unfortunately, computing an optimal tree decomposition is NP-hard~\cite{ArnborgCP87}. As many of the algorithmic applications of treewidth require a tree decomposition to be able to work, there has been long record of algorithms computing optimal~\cite{ArnborgCP87,Bodlaender96,BodlaenderK96,Reed92} or near-optimal~\cite{BelbasiF22,BodlaenderDDFLP16,Korhonen21} tree decompositions with no end in sight~\cite{KorhonenL22}, as well as a long series of experimental work on heuristically computing good tree decompositions~\cite{BodlaenderK10,BodlaenderK11,DellHJKKR16,DellKTW18}. In this paper, we present a new fixed-parameter tractable approximation algorithm for the notion of $\hh$-treewidth, a generalization of treewidth which has recently attracted significant attention~\cite{AgrawalKLPRSZ22,DBLP:conf/mfcs/EibenGHK19,JansenK21,JansenKW21}. Before describing our contributions for $\hh$-treewidth, we summarize the most important background to motivate the problem.

The popularity of treewidth as a graph parameter can be attributed to the fact that it has very good algorithmic properties (by Courcelle's theorem, any problem that can be formulated in Counting Monadic Second-Order (CMSO$_2$) logic can be solved in linear time on graphs of bounded treewidth~\cite{CourcelleE12}), while also having a very elegant mathematical structure theory. Unfortunately, simple substructures like grids or cliques in a graph can already make its treewidth large. This means that for many input graphs of interest, the treewidth is too large for an approach based on treewidth to be efficient: the running time of many treewidth-based algorithms are of the form~$f(k) \cdot n^{\Oh(1)}$, where~$f$ is an exponential function in the treewidth~$k$ and~$n$ is the total number of vertices of the graph.

Several approaches have been taken to cope with the fact that treewidth is large on graphs with large cliques or large induced grid subgraphs. One approach lies in generalized width measures like cliquewidth or rankwidth~\cite{OumS06}, by essentially replacing the use of separations of small order (which are encoded in tree decompositions), by separations of large order but in which the interactions between the two sides is well-structured. Unfortunately this generality comes at a price in terms of algorithmic applications~\cite{FominGLS09,FominGLS10,FominGLS14}.

This has recently led Eiben, Ganian, Hamm, and Kwon~\cite{DBLP:conf/mfcs/EibenGHK19} to enrich the notion of treewidth in a different way. Consider a hereditary class $\hh$ of graphs, such as bipartite graphs. The notion of $\hh$-treewidth aims to capture how well a graph~$G$ can be decomposed into subgraphs belonging to~$\hh$ which only interact with the rest of the graph via small vertex sets which are organized in a tree-like manner. While we defer formal definitions of $\hh$-treewidth to Section~\ref{sec:prelims}, an intuitive way to think of the concept is the following: a graph $G$ has $\hh$-treewidth at most~$k$ if and only if it can be obtained from a graph $G_0$ with a tree decomposition of width at most~$k$ by the following process: repeatedly insert a subgraph $H_i$ belonging to graph class $\hh$, such that the neighbors of $H_i$ in the rest of the graph are all contained in a single bag of the tree decomposition of $G_0$. The $\hh$-subgraphs~$H_i$ inserted during this process are called \emph{base components} and their neighborhoods have size at most~$k+1$. When $\hh$ is a graph class of unbounded treewidth, like bipartite graphs, the $\hh$-treewidth of a graph can be arbitrarily much smaller than its treewidth. This prompted an investigation of the algorithmic applications of $\hh$-treewidth.

In recent works~\cite{AgrawalKLPRSZ22,DBLP:conf/mfcs/EibenGHK19,JansenKW21}, the notion of $\hh$-treewidth was used to develop new algorithms to solve vertex-deletion problems. Many classic NP-hard problems in algorithmic graph theory can be phrased in the framework of \hhdel{}: find a minimum vertex-subset~$S$ of the input graph~$G$ such that~$G-S$ belongs to a prescribed graph class~$\hh$. Examples include \textsc{Vertex Cover} (where~$\hh$ is the class of edgeless graphs), \textsc{Odd Cycle Transversal} (bipartite graphs), and \textsc{Vertex Planarization} (planar graphs). All these problems are known to be fixed-parameter tractable~\cite{ChenKX10,DBLP:conf/soda/JansenLS14,Kawarabayashi09,MarxS12,ReedSV04} when parameterized by the size of a desired solution: there are algorithms that, given an $n$-vertex graph~$G$ and integer~$k$, run in time~$f(k) \cdot n^{\Oh(1)}$ and output a vertex set~$S \subseteq V(G)$ of size at most~$k$ for which~$G - S \in \hh$, if such a set exists. These algorithms show that large instances whose optimal solutions are small, can still be solved efficiently. Alternatively, since the mentioned graph classes~$\hh$ can be defined in CMSO$_2$, these vertex-deletion problems can be solved in time~$f(w) \cdot n$ parameterized by the treewidth~$w$ of the input graph via Courcelle's theorem, which shows that instances of small treewidth (but whose optimal solutions may be large) can be solved efficiently.

The notion of $\hh$-treewidth (abbreviated as $\hhtw$ from now on) can be used to combine the best of both worlds. It is not difficult to show that if a graph~$G$ has a vertex set~$S$ of size~$k$ for which~$G-S \in \hh$ (we call such a set an $\hh$-deletion set), then the $\hh$-treewidth of~$G$ is at most~$k$: simply take a trivial tree decomposition~$(T,\chi)$ consisting of a single bag of size~$k$ for the graph~$G_0 := G[S]$, so that afterwards the graph~$G$ can be obtained from~$G_0$ by inserting the graph~$H = G-S$, which belongs to~$\hh$ and has all its neighbors in a single bag of~$(T,\chi)$. Since the $\hh$-treewidth of~$G$ is also never larger than its standard treewidth, the \emph{hybrid} (cf.~\cite{AgrawalR22}) parameterization by $\hh$-treewidth dominates both the parameterizations by the solution size and the treewidth of the graph. This raises the question whether existing fixed-parameter tractability results for parameterizations of \hhdel{} by treewidth or solution size, can be extended to $\hhtw$.

It was recently shown~\cite{AgrawalKLPRSZ22} that when it comes to \emph{non-uniform} fixed-parameter tractability characterizations, the answer to this question is positive. If~$\hh$ satisfies certain mild conditions, which is the case for \bmp{all graph classes mentioned so far}, then for each value of~$k$ there exists an algorithm~$\mathcal{A}_{\hh,k}$ that, given a graph~$G$ with~$\hhtw(G) \leq k$ and target value~$t$, decides whether or not~$G$ has an $\hh$-deletion set of size at most~$t$. There is a constant~$c_\hh$ such that each algorithm~$\mathcal{A}_{\hh,k}$ runs in time~$\Oh(n^{c_\hh})$, so that the overall running time can be bounded by~$f(k) \cdot n^{c_\hh}$; however, no bounds on the function~$f$ are given and in general it is unknown how to construct the algorithms whose existence is proven. \bmp{Another recent paper}~\cite{JansenK21} gave concrete FPT algorithms to solve $\hhdel{}$ parameterized by $\hhtw$ for certain cases of~$\hh$, including the three mentioned ones. For example, it presents an algorithm that solves \textsc{Odd Cycle Transversal} in time~$2^{\Oh(k^3)} \cdot n^{\Oh(1)}$, parameterized by $\mathsf{bipartite}$-treewidth. The bottleneck in the latter approach lies in the computation of a suitable tree $\hh$-decomposition: on a graph of~$\hhtw(G) \leq k$, the algorithm runs in~$2^{\Oh(k \log k)} \cdot n^{\Oh(1)}$ time to compute a tree $\hh$-decomposition of width~$w \in \Oh(k^3)$, and then optimally solves \textsc{Odd Cycle Transversal} on the decomposition of width~$w$ in time~$2^{\Oh(w)} \cdot n^{\Oh(1)} \leq 2^{\Oh(k^3)} \cdot n^{\Oh(1)}$. Note that the parameter dependence of this algorithm is much worse than for the parameterizations by solution size and by treewidth, both of which can be solved in single-exponential time~$2^{\Oh(k)} \cdot n^{\Oh(1)}$~\cite{ReedSV04,LokshtanovMS18}. To improve the running times of algorithms for \hhdel based on hybrid parameterizations, improved algorithms are therefore required to compute approximate tree $\hh$-decompositions. These form the subject of our work.

\subparagraph{Our contribution: $\hh$-treewidth.} 

We develop generic FPT algorithms to approximate $\hh$-treewidth, for graph classes $\hh$ which are hereditary and closed under taking the disjoint union of graphs. To approximate $\hh$-treewidth, all our algorithm needs is access to an oracle for solving \hhdel{} parameterized by solution size. The values of the solution size for which the oracle is invoked, will be at most twice as large as the $\hh$-treewidth of the graph we are decomposing. Hence existing algorithms for solution-size parameterizations of \hhdel{} can be used as a black box to form the oracle. Aside from the oracle calls, our algorithm only takes~$8^k\cdot kn(n+m)$ time on an $n$-vertex graph with~$m$ edges. So whenever the solution-size parameterization can be solved in single-exponential time, an approximate tree $\hh$-decomposition can be found in single-exponential time. The approximation factor of the algorithm is~$5$, which is a significant improvement over earlier $\mathsf{poly}(\mathsf{opt})$ approximations running in superexponential time. The formal statement of our main result is the following.

\begin{restatable}{theorem}{restTreewidthAlgorithm}
\label{thm:intro:treewidth}
Let $\hh$ be a hereditary and union-closed class of graphs.
There is an algorithm that,
using oracle-access to an algorithm~$\mathcal{A}$ for \hhdel,
takes as input an $n$-vertex $m$-edge graph $G$, integer $k$, and either computes a tree $\hh$-decomposition of $G$ of width at most $5k+5$ \bmp{consisting of~$\Oh(n)$ nodes,} or correctly concludes that $\hhtw(G) > k$.
The algorithm runs in time $\Oh(8^k\cdot kn(n+m))$, polynomial space, and makes $\Oh(8^k n)$ calls to $\mathcal{A}$ on induced subgraphs of $G$ and parameter $2k+2$.
\end{restatable}

\cref{thm:intro:treewidth} yields the first constant-factor approximation algorithms for $\hhtw$ that run in single-exponential time. For example, for $\hh$ the class of bipartite graphs the running time becomes~$\Oh(72^k\cdot n^2(n+m))$, and for interval graphs we obtain~$\Oh(8^{3k}\cdot n(n+m))$ (\cref{cor:htw:apx} lists results for more classes $\hh$). Combining these approximate decompositions with existing algorithms that solve \hhdel{} on a given tree $\hh$-decomposition, we obtain ETH-tight algorithms as a consequence. \textsc{Odd Cycle Transversal} can be solved in time~$2^{\Oh(k)} \cdot n^{\Oh(1)}$, and \textsc{Vertex Planarization} can be solved in time~$2^{\Oh(k \log k)} \cdot n^{\Oh(1)}$ when parameterized by~$\hhtw$ for~$\hh$ the class of bipartite and planar graphs, respectively, without having to supply a decomposition in the input. For \textsc{Vertex Planarization}, the previous-best bound~\cite{JansenKW21-arxiv} was~$2^{\Oh(k^5 \log k)} \cdot n^{\Oh(1)}$. Note that for the planarization problem, a parameter dependence of~$2^{o(k \log k)}$ is impossible assuming the Exponential Time Hypothesis; this already holds for the larger parameterization by treewidth~\cite{Pilipczuk17}. For \textsc{Odd Cycle Transversal}, an algorithm running in time~$2^{o(n)}$ would violate the Exponential Time Hypothesis, which follows by a simple reduction from \textsc{Vertex Cover} for which such a lower bound is known~\cite[Theorem 14.6]{CyganFKLMPPS15}. This implies that the solution size parameterization cannot be solved in subexponential~time.

Compared to existing algorithms to approximate treewidth, the main obstacle we have to overcome in \cref{thm:intro:treewidth} is identifying the base components of an approximate decomposition in a suitable way. The earlier FPT \bmp{$\mathsf{poly}(\mathsf{opt})$-approximation} for $\hhtw$ effectively reduced the input graph~$G$ to a graph~$G'$ by repeatedly extracting large~$\hh$-subgraphs with small neighborhoods, in such a way that the treewidth of~$G'$ can be bounded in terms of~$\hhtw(G)$, while a tree decomposition of~$G'$ can be lifted into an approximate tree $\hh$-decomposition of~$G$. Several steps in this process led to losses in the approximation factor. \bmp{To obtain our 5-approximation, we} avoid the translation between~$G$ and~$G'$, and work directly on decomposing the input graph~$G$. 

{Our recursive decomposition algorithm works similarly as the Robertson-Seymour 4-approximation algorithm for treewidth~\cite{RobertsonS95} (cf.~\cite[\S 7.6]{CyganFKLMPPS15}). When given a graph~$G$ and integer~$k$ with~$\hhtw(G) \leq k$, the algorithm maintains a vertex set~$S$ of size~$3k+4$ which forms the boundary between the part of the graph that has already been decomposed and the part that still needs to be processed. If~$S$ has a $\frac{2}{3}$-balanced separator~$R$ of size~$k+1$, we can proceed in the usual way: we split the graph based on~$R$, recursively decompose the resulting parts, and combine these decompositions by adding a bag containing~$R \cup S$ as the root. 
\bmp{If~$S$ does not have a balanced separator of size~$k+1$, then we show (modulo some technical details) that for any optimal tree $\hh$-decomposition, there is a subset~$S' \subseteq S$ of~$2k+3$ vertices which belong to a single base component~$H_0$.}
%
Our main insight is that such a set~$S'$ can be used in a win/win approach, by maintaining an $\hh$-deletion set~$X$ during the decomposition process that initially contains all vertices. To make progress in the recursion, we would like to split off a base component containing~$S'$ via a separator~$U$ of size at most~$2k+2$, while adding~$U$ to the boundary of the remainder of the graph to be decomposed. To identify an induced $\hh$-subgraph with small neighborhood that can serve as a base component, we compute a minimum~$(S',X)$-separator~$U$ \mic{(we allow $U$ to intersect the sets $S', X$)}.
Any connected component~$H$ of~$G-U$ that contains a vertex from~$S'$ does not contain any vertex of the $\hh$-deletion set~$X$, so~$H$ is an induced subgraph of~$G-X$ which implies~$H \in \hh$ for hereditary~$\hh$. Hence if there is an~$(S',X)$-separator~$U$ of size at most~$2k+2$, we can use it to split off base components neighboring~$U$ that eliminate~$2k+3$ vertices from~$S$ from the boundary, thereby making room to insert~$U$ into the boundary without blowing up its size. Of course, it may be that all~$(S',X)$-separators are larger than~$2k+2$; by Menger's theorem, this happens exactly when there is a family~$\mathcal{P}$ of~$2k+3$ vertex-disjoint~$(S',X)$-paths. Only~$k+1$ paths in~$\mathcal{P}$ can \emph{escape} the base component~$H_0$ covering~$S'$ since its neighborhood has size at most~$k+1$, so that~$k+2$ of them end in a vertex of the deletion set~$X$ that lies in~$H_0$. The key point is now that this situation implies that~$X$ is redundant in a technical sense: if we let~$X'$ denote the endpoints of~$k+2$ $(S',X)$-paths starting and ending in~$H_0$, we can obtain a smaller $\hh$-deletion set by replacing~$X'$ by the neighborhood of~$H_0$, which has size at most~$k+1$. This replacement is valid as long as~$\hh$ is hereditary and union-closed. Using an oracle for \hhdel{} parameterized by solution size, we can therefore efficiently find a smaller~$\hh$-deletion set when \bmp{we know}~$X'$. \bmp{While the algorithm does not know~$X'$ in general, this type of argument leads to the win/win:} either there is a small~$(S',X)$-separator which we can use to split off a base component, or there is a large family of vertex-disjoint~$(S',X)$-paths which allows the $\hh$-deletion set to be improved. As the latter can only happen~$|V(G)|$ times, we must eventually identify a base component to split off, allowing the recursion to proceed.}


\subparagraph{Our contribution: $\hh$-elimination distance.} 
The $\hh$-elimination distance~$\hhdepth(G)$ of a graph~$G$ is a parameter~\cite{DBLP:journals/algorithmica/BulianD16,DBLP:journals/algorithmica/BulianD17} that extends treedepth~\cite{NesetrilM12} similarly to how $\hh$-treewidth extends treewidth. For hereditary and union-closed classes~$\hh$, the $\hh$-elimination distance of a graph~$G$ is the minimum number of rounds needed to turn~$G$ into a member of~$\hh$, when a round consists of removing one vertex from each connected component. Such an elimination process can be represented by a tree structure called \emph{$\hh$-elimination forest}. Aside from the fact that computing the $\hh$-elimination distance may reveal interesting properties of a graph~$G$, a second motivation for studying this parameter is that it can facilitate \emph{polynomial-space} algorithms for solving \hhdel{}, while the parameterization by $\hhtw$ (which is never larger) typically gives rise to exponential-space algorithms. At a high level, the state of the art for computing $\hhdepth$ is similar as for $\hhtw$: there is an exact non-uniform FPT algorithm with unspecified parameter dependence that works as long as $\hh$ satisfies some mild conditions~\cite{AgrawalKLPRSZ22}, while uniform $\mathsf{poly}(\mathsf{opt})$-approximation algorithms running in time~$2^{k^{\Oh(1)}} \cdot n^{\Oh(1)}$ are known for several concrete graph classes~$\hh$~\cite{JansenKW21}. 

By leveraging similar ideas as for \cref{thm:intro:treewidth}, we also obtain improved FPT-approximation algorithms for $\hhdepth$.
The following theorem gives algorithms for two settings: one for an algorithm using polynomial space whenever the algorithm $\mathcal{A}$ for \hhdel does, which is the case for most of the considered graph classes, and one for an exponential-space algorithm with a better approximation ratio.
\newpage

\begin{restatable}{theorem}{restTreedepthAlgorithm}
\label{thm:intro:elimination-distance}
Let $\hh$ be a hereditary and union-closed class of graphs.
There exists an algorithm that, using oracle-access to an algorithm~$\mathcal{A}$ for \hhdel, takes as input an $n$-vertex  graph $G$ and integer $k$, runs in time $n^{\Oh(1)}$,
makes $n^{\Oh(1)}$ calls to $\mathcal{A}$ on induced subgraphs of $G$ and parameter $2k$, and
either concludes that $\hhdepth(G) > k$ or outputs an $\hh$-elimination forest of depth $\Oh(k^3 \log^{3/2}k)$. 

Under the same assumptions, there is an algorithm that runs in time $2^{\Oh(k^2)}\cdot n^{\Oh(1)}$, makes $n^{\Oh(1)}$ calls to $\mathcal{A}$ on induced subgraphs of $G$ and parameter $2k$, and
either concludes that $\hhdepth(G) > k$ or outputs an $\hh$-elimination forest of depth $\Oh(k^2)$.
\end{restatable}

In the previous work~\cite{JansenKW21, JansenKW21-arxiv} such a dependence on $k$ was possible only in two cases: when $\hh$ is the class of bipartite graphs or when $\hh$ is defined by a finite family of forbidden induced subgraphs.
In the general case, our result effectively shaves off a single $k$-factor in the depth of a returned decomposition \bmp{and in the exponent of the running time}, compared to the previously known approximations.
\cref{thm:intro:elimination-distance} entails better approximation algorithms for  $\hh$-elimination distance for classes of e.g. chordal, interval, planar, bipartite permutation, or distance-hereditary graphs.

\subparagraph{Organization.} The remainder of the paper is organized as follows. We continue by presenting
formal preliminaries in Section~\ref{sec:prelims}. In Section~\ref{sec:treewidth} we treat $\hh$-treewidth, developing the theory and subroutines needed to prove \cref{thm:intro:treewidth}.
The list of its applications for concrete graph classes is provided in \cref{sec:applications}. The proof of \cref{thm:intro:elimination-distance} is presented in \cref{sec:treedepth}. 
We conclude in \cref{sec:conclusion}.

\section{Preliminaries} \label{sec:prelims}

\subparagraph{Graphs and graph classes.}
We consider \bmp{finite}, simple, undirected graphs. \bmp{We denote} the vertex and edge sets of a graph\bmp{~$G$}  by $V(G)$ and $E(G)$ respectively, with $|V(G)| = n$ and $|E(G)| = m$. For a set of vertices $S \subseteq V(G)$, by $G[S]$ we denote the graph induced by $S$. We use shorthand $G-v$ and $G-S$ for $G[V(G) \setminus \{v\}]$ and $G[V(G) \setminus S]$, respectively. The open neighborhood $N_G(v)$ of $v \in V(G)$ is defined as $\{u \in V(G) \mid uv \in E(G)\}$. The closed neighborhood of~$v$ \bmp{is} $N_G[v] = N_G(v) \cup \{v\}$. For $S \subseteq V(G)$, we have $N_G[S] = \bigcup_{v \in S} N_G[v]$ and $N_G(S) = N_G[S] \setminus S$.
\mic{We define the boundary $\partial_G(S)$ of the vertex set $S$ as $N_G(V(G) \sm S)$, \bmp{i.e., those vertices of~$S$ which have a neighbor outside~$S$}.}

A class of graphs $\hh$ is called \emph{hereditary} if for \bmp{any} $G \in \hh$, every induced subgraph of $G$ also belongs to $\hh$.
Furthermore, $\hh$ is \emph{union-closed} if for \bmp{any} $G_1, G_2 \in \hh$ the \bmp{disjoint union} of $G_1$ and $G_2$ also belongs to $\hh$.
For a graph class $\hh$ and a graph $G$, a set $X \subseteq V(G)$ is called an $\hh$-deletion set \bmp{in~$G$} if $G-X \in \hh$.
For a graph class $\hh$, the parameterized problem \hhdel takes a graph $G$ and parameter $k$ \bmp{as input}, and either outputs \bmp{a} minimum-size $\hh$-deletion set in $G$ or reports that there is no such set of size at most $k$. 

\subparagraph{Separators.}

For two (not necessarily disjoint) sets $X,Y \subseteq V(G)$ in a graph~$G$, a set $P \subseteq V(G)$ is an $(X,Y)$-separator if no connected component of $G-P$ contains a vertex from both \bmp{$X \setminus P$} and \bmp{$Y \setminus P$}. Such a separator may intersect $X \cup Y$. \bmp{Equivalently,~$P$ is an~$(X,Y)$-separator} if each $(X,Y)$-path contains a vertex of $P$. The minimum cardinality of such a separator is denoted $\lambda_G(X,Y)$. \bmp{By Menger's theorem,~$\lambda_G(X,Y)$ is equal to the maximum cardinality of a set of pairwise vertex-disjoint $(X,Y)$-paths.}
A pair $(A,B)$ of subsets of $V(G)$ is a \emph{separation} in~$G$ if $A \cup B = V(G)$ and $G$ has no edges between $A \setminus B$ and $B \setminus A$. \jjh{Its order is defined as $|A \cap B|$.}

\begin{observation} \label{obs:separation:lambda}
For two sets $X,Y \subseteq V(G)$, it holds that $\lambda_G(X,Y) \le k$ if and only if there exists a separation $(A,B)$ in $V(G)$ such that $X \subseteq A$, $Y \subseteq B$, and {$|A \cap B| \le k$}.
\end{observation}


\bmp{The following theorem summarizes how, given vertex sets~$X, Y \subseteq V(G)$ and a bound~$k$, we can algorithmically find a small-order separation or \mic{a} large system of vertex-disjoint paths. The statement follows from the analysis of the Ford-Fulkerson algorithm for maximum~$(X,Y)$-flow in which each vertex has a capacity of~$1$.  If the algorithm has not terminated within~$k$ iterations, then the flow of value~$k+1$ yields~$k+1$ vertex-disjoint paths. If it terminates earlier, a suitable separation can be identified based on reachability in the residual network of the last iteration.}

\begin{theorem}[Ford-Fulkerson, {see \cite[Thm. 8.2]{CyganFKLMPPS15} and \cite[\S 9.2]{Schrijver03}}]\label{cor:subset:ford}
There is an algorithm that, given an $n$-vertex $m$-edge graph $G$, sets $X,Y \subseteq V(G)$, and integer $k$, runs in time $\Oh(k(n+m))$ and determines whether $\lambda_G(X,Y) \le k$.
If so, the algorithm also returns \bmp{a separation~$(A,B)$ in~$G$ with~$X \subseteq A$,~$Y\subseteq B$, and~$|A \cap B| \leq k$}. 
\mic{Otherwise, the algorithm returns a family of $k+1$ vertex-disjoint $(X,Y)$-paths.}
\end{theorem}

\subparagraph{\hhtwfull{}.} 
\mic{We continue by giving a formal definition
of a tree $\hh$-decomposition.}

\begin{definition} \label{def:tree:h:decomp}
For a graph class $\hh$, a tree $\hh$-decomposition of graph $G$ is a triple $(T, \chi, L)$ where~$L \subseteq V(G)$,~$T$ is a rooted tree, and~$\chi \colon V(T) \to 2^{V(G)}$, such that:
\begin{enumerate}
    \item For each~$v \in V(G)$ the nodes~$\{t \mid v \in \chi(t)\}$ form a {non-empty} connected subtree of~$T$. \label{item:tree:h:decomp:connected}
    \item For each edge~$uv \in E(G)$ there is a node~$t \in V(T)$ with~$\{u,v\} \subseteq \chi(t)$.
    \item For each vertex~$v \in L$, there is a unique~$t \in V(T)$ with~$v \in \chi(t)$, and~$t$ is a leaf of~$T$.\label{item:tree:h:decomp:unique}
    \item For each node~$t \in V(T)$, the graph~$G[\chi(t) \cap L]$ belongs to~$\hh$. \label{item:tree:h:decomp:base}
\end{enumerate}

{The \emph{width} of a tree $\hh$-decomposition is defined as~$\max(0, \max_{t \in V(T)} |\chi(t) \setminus L| - 1)$.} The $\hh$-treewidth of a graph~$G$, denoted~$\hhtw(G)$, is the minimum width of a tree $\hh$-decomposition of~$G$.
The connected components of $G[L]$ are called base components.

{A pair~$(T, \chi)$ is a (standard) \emph{tree decomposition} if~$(T, \chi, \emptyset)$ satisfies all conditions of an $\hh$-decomposition; the choice of~$\hh$ is irrelevant.}
\end{definition}

For a rooted tree decomposition $(T,\chi)$, $T_t$ denotes the subtree of $T$ rooted at $t \in V(T)$, while $\chi(T_t) = \bigcup_{x \in V(T_t)} \chi(x)$. Similarly as treewidth, $\hh$-treewidth is a monotone parameter with respect to taking induced subgraphs.

\begin{observation}\label{obs:tree:induced}
Let $\hh$ be a hereditary class of graphs, $G$ be a graph, and $H$ be an induced subgraph of $G$.
Then $\hhtw(H) \le \hhtw(G)$.
\end{observation}

\section{\texorpdfstring{Approximating $\hh$-treewidth}{Approximating H-treewidth}} \label{sec:treewidth}

\mic{We make preparations for the proof of \cref{thm:intro:treewidth}.
First, we formalize the concept of a~potential base component using the notion of an $(\hh,\ell)$-{separation} and relate it to {\em redundant} subsets in a solution to \hhdel{}.
Next, we prove a counterpart of the balanced-separation property for graphs of bounded \hhtwfull{} and explain how it allows us to apply a win/win approach in a single step of the decomposition algorithm.}

\subsection{\texorpdfstring{Redundancy and $(\hh,\ell)$-separations}{Redundancy and (H,l)-separations}}

We summon the following concept from the previous work on $\hh$-treewidth~\cite{JansenKW21} \bmp{to capture $\hh$-subgraphs with small neighborhoods.}

\begin{definition} \label{def:h:ell:sep}
For disjoint $C, S \subseteq V(G)$, the pair $(C,S)$ is called an $(\hh,\ell)$-{separation} in~$G$~if (1) $G[C] \in \hh$, (2) $|S| \le \ell$, and (3) $N_G(C) \subseteq S$.
\end{definition}

\bmp{This notion is tightly connected to the base components of tree $\hh$-decompositions. For any tree $\hh$-decomposition~$(T,\chi,L)$ of \mic{width~$k$} of a graph~$G$, for any node~$t \in T$, the graph~$G[\chi(t) \cap L]$ belongs to~$\hh$ so that~$C := \chi(t) \cap L$ satisfies \cref{def:h:ell:sep}. The open neighborhood of~$\chi(t) \cap L$ is a subset of~$S := \chi(t) \setminus L$, which follows from the fact that vertices of~$L$ only occur in a single bag, while each edge has both endpoints covered by a single bag. Since~$|S| \leq k+1$ by definition of the width of a tree $\hh$-decomposition, this leads to the following observation.}

\begin{observation}\label{obs:tree:secluded}
Let $(T, \chi, L)$ be a tree $\hh$-decomposition of a graph $G$ of \mic{width~$k$}. \bmp{For each node~$t \in V(T)$, 
the pair~$(\chi(t) \cap L, \chi(t) \setminus L)$ is an~$(\hh,k+1)$-separation in~$G$.}
\end{observation}

The following concept will be useful when working with~$(\hh,\ell)$-separations.

\begin{definition}
For an \hsepk $(C, S)$ and set $Z \subseteq V(G)$,
we say that $(C,S)$ {covers} $Z$ if $Z \subseteq C$, \bmp{or {weakly covers} $Z$ if $Z \subseteq C \cup S$.}
Set $Z \subseteq V(G)$ is called \bmp{(weakly)} $(\hh, \ell)$-{separable} if
there exists an $(\hh, \ell)$-{separation} that \bmp{(weakly)} covers $Z$.
\end{definition}

\mic{We introduce both notions to  keep consistency with the earlier work~\cite{JansenKW21} but in fact we will be interested only in weak coverings.}
Following the example above, the set $Z = \chi(t)$ is weakly $(\hh, k+1)$-{separable} but not necessarily $(\hh, k+1)$-{separable}.

Next, we introduce the notion of redundancy for solutions to \hhdel. 

\begin{definition}
 For an $\hh$-deletion set $X$ in $G$ we say that a subset $X' \subseteq X$ is \emph{redundant} in $X$ if there exists a set $X'' \subseteq V(G)$ smaller than $|X'|$ such that $(X \setminus X') \cup X''$ is also an $\hh$-deletion set in $G$.    
\end{definition} 

\mic{We remark that redundancy has been studied in the context of local-search strategies (cf.~\cite{GuoHNS13,GuoHK14}).}
\bmp{It is known that for \textsc{Vertex Cover} finding a redundant subset $X'$ in a solution $X$ is FPT in graphs of bounded local treewidth but W[1]-hard in general, when parameterized by the size of $|X'|$~\cite{FellowsFLRSV12}.}
However, when $X'$ is given, one can easily check whether it is redundant using an algorithm for \hhdel parameterized by the solution size, due to the following observation.

\begin{observation}
 Let $X$ be an $\hh$-deletion set in a graph $G$.
 A subset $X' \subseteq X$ is redundant in $X$ if and only if the graph $G - (X \sm X')$ has an  $\hh$-deletion set smaller than $|X'|$.
\end{observation}

An important observation is that when $X' \subseteq X$ of size at least $\ell + 1$ is weakly $(\hh,\ell)$-separable, then it is redundant in $X$ by a simple exchange argument.
This fact has been already leveraged in \bmp{previous work}~\cite{AgrawalKLPRSZ22, JansenKW21} when analyzing the structure of minimum-size $\hh$-deletion sets, which clearly cannot contain any redundant subsets.
We exploit it in a different context, to prove that if
there is a large flow between an $\hh$-deletion set $X$ and
a weakly $(\hh,\ell)$-separable set $Z$, then $X$ has a redundant subset.
\mic{Subsequently, we will show that this redundant subset
can efficiently \jjh{be} detected.}

\begin{figure}
    \centering
    \includegraphics[scale=1.2]{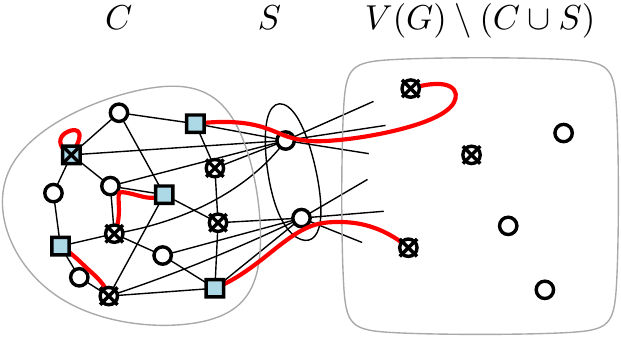}
    \caption{
    \mic{Illustration for \cref{lem:subset:redundant}:} an~$(\hh,\ell)$-separation \mic{$(C,S)$} in a graph~$G$ where~$\hh$ is the class of triangle-free graphs and~$\ell=2$. Vertices marked with a cross form an $\hh$-deletion set~$X$ in~$G$, while the \mic{set~$Z$} of size~$2\ell+1=5$ marked with blue squares is weakly~$(\hh,\ell)$-separable. A set of~$2\ell+1$ vertex-disjoint~$(Z,X)$-paths~$\mathcal{P}$ is highlighted, witnessing~$\lambda_G(Z,X)=|Z|$.  Since~$|X \cap C| > \ell$, the set~$X$ is not a minimum $\hh$-deletion set in~$G$: it can be improved by replacing~$X \cap C$ with~$S$. When~$(C,S)$ weakly covering~$Z$ exists but is unknown to the algorithm, we can \mic{still} improve~$X$ since the set of $X$-endpoints of~$\mathcal{P}$ \mic{also} form a redundant set in~$X$.}
    \label{fig:separationImprovement}
\end{figure}

\begin{lemma}
\label{lem:subset:redundant}
Let $\hh$ be a hereditary and union-closed class of graphs.
Consider a graph $G$, an $\hh$-deletion set $X$ in $G$, and \bmp{a weakly} $(\hh, \ell)$-{separable} set $Z \subseteq V(G)$.
Suppose that there exists a subset $X' \subseteq X$ of size $2\ell+1$ such that $\lambda_G(Z,X') = 2\ell+1$.
Then $X'$ is redundant in~$X$.
\end{lemma}
\begin{proof}
\mic{Let $(C,S)$ be an~$(\hh,\ell)$ separation in~$G$ with~$Z \subseteq C \cup S$.}
From the definition of \mic{an}~$(\hh,\ell)$-separation, we have~$G[C] \in \hh$ while~$N_G(C) \subseteq S$ and~$|S|\leq \ell$. 

By Menger's theorem, the cardinality of a maximum packing of vertex-disjoint~$(Z,X')$-paths equals~$\lambda_G(Z,X')$. Hence there exists a family~$\mathcal{P} = \{P_1, \ldots, P_{2\ell+1}\}$ of vertex-disjoint paths, each of which connects a unique vertex~$z_i \in Z$ to a unique vertex~$x_i \in X'$ (possibly $x_i = z_i$). At most~$|S|$ of these paths intersect the separator~$S$ of the~$(\hh,\ell)$-separation \bmp{(see Figure \ref{fig:separationImprovement})}. Let~$X'' = \{x_i \mid P_i \cap S = \emptyset\}$ denote the $X'$-endpoints of those paths not intersecting~$S$, and let~$\mathcal{P}'$ be the corresponding paths. Each path~$P_i$ in~$\mathcal{P}'$ is disjoint from~$S$ and has an endpoint~$z_i \in Z$. Since~$Z \subseteq C \cup S$, the~$z_i$ endpoint belongs to~$C$. As~$N_G(C) \subseteq S$ and~$P_i$ does not intersect~$S$, the other endpoint~$x_i$ also belongs to~$C$. Hence all vertices of~$X''$ belong to~$C$, and there are at least~$2\ell+1 - \ell = \ell+1$ of them.

Let~$X^* := (X \setminus X'') \cup S$, and observe that~$|X^*| < |X|$ since~$|X''| \geq \ell+1$ while~$|S|\leq \ell$. We prove that~$G - X^* \in \hh$, by showing that~$S$ is an $\hh$-deletion set in~$G - (X \setminus X'')$. Since~$\hh$ is union-closed, it suffices to argue that each connected component~$H$ of~$G - ((X \setminus X'') \cup S)$ belongs to~$\hh$. If~\bmp{$H$} contains no vertex of~$X''$, then~\bmp{$H$} is an induced subgraph of~$G - X \in \hh$ and therefore~\bmp{$H \in \hh$} since the graph class is hereditary. If~\bmp{$H$} contains a vertex of~$X'' \subseteq C$, then the component~$H$ is an induced subgraph of~$G[C]$ since~$N_G(C) \subseteq S$ \bmp{is part of the set~$X^*$}. Hence~$H$ is an induced subgraph of~$G[C] \in \hh$, which implies~$H \in \hh$ as~$\hh$ is hereditary. This shows that~$X^*$ is indeed an $\hh$-deletion set.

Since~$(X \setminus X'') \cup S$ is an $\hh$-deletion set smaller than~$X$, the set~$X''$ is redundant in~$X$. As~$X' \supseteq X''$, it follows that~$X'$ is redundant as well.
\end{proof}

\subsection{The win/win strategy}

\mic{The classic 4-approximation algorithm for computing a (standard) tree decomposition is based on the existence of balanced separators in graphs of bounded treewidth.}
In a graph $G$ of treewidth $\le k$, any set $S$ of $3k+4$ vertices 
can by partitioned into $S = S_A \cup S_B$ in such a way that
$|S_A|, |S_B| \le 2k+2$ and $\lambda_G(S_A,S_B) \le k+1$~\cite[Corollary 7.21]{CyganFKLMPPS15}.
This is not always possible if we only have a bound on $\hh$-treewidth $\hhtw(G) \le k$ because a large subset $S'$ of $S$ might lie in a single \bmp{well-connected} base component of a tree $\hh$-decomposition, i.e., a base component whose standard treewidth is large.
But then $S'$ is weakly $(\hh,k+1)$-separable, \bmp{which can also} be exploited when constructing a decomposition.
We show that this is in fact the only scenario in which we cannot split $S$ in a balanced way.

\begin{lemma}
\label{lem:split:or:separable}
Let $\hh$ be a hereditary and union-closed class of graphs.
Let $G$ be a graph with $\hhtw(G) \le k$. For any set $S \subseteq V(G)$ of size $3k+4$, at least one of the following holds.
\begin{enumerate}
     \item There is a partition $S = S_A \cup S_B$ such that $|S_A|, |S_B| \le 2k+2$ and $\lambda_G(S_A,S_B) \le k+1$.\label{case:split}
     \item There is a set $S' \subseteq S$ of size $2k+3$ which is weakly $(\hh,k+1)$-separable.\label{case:separable}
\end{enumerate}
\end{lemma}
\begin{proof}
Consider an optimal tree $\hh$-decomposition~$(T,\chi,L)$ of~$G$, so that~$|\chi(t) \setminus L| \leq k+1$ for each~$t \in V(T)$. \bmp{Let~$r \in V(T)$ be its root}. We start by showing that~\eqref{case:separable} holds if some leaf bag of the decomposition contains~$2k+3$ vertices from~$S$.

So suppose there exists a leaf~$t \in V(T)$ with~$|\chi(t) \cap S| \geq 2k+3$, and let~$S' \subseteq \chi(t) \cap S$ be an arbitrary subset of size exactly~$2k+3$. 
\mic{\cref{obs:tree:secluded}} ensures that~$(C^* := \chi(t) \cap L, S^* := \chi(t) \setminus L)$ is an~$(\hh,k+1)$-separation, which weakly covers~$\chi(t)$ and therefore~$S'$. Hence~\eqref{case:separable} holds.

In the remainder, it suffices to show that~\eqref{case:split} holds when there is no leaf~$t \in V(T)$ with~$|\chi(t) \cap S| \geq 2k+3$. Pick a deepest node~$t^*$ in the rooted tree~$T$ for which~$|S \cap \chi(T_{t^*})| \geq 2k+3$. Then~$t^*$ is not a leaf since the previous case did not apply, so by definition of tree $\hh$-decomposition we have~$\chi(t^*) \cap L = \emptyset$. Let~$D_1, \ldots, D_p$ be the connected components of~$G - \chi(t^*)$. Since the pair~$(T,\chi)$ satisfies all properties of a standard tree decomposition, the bag~$\chi(t^*)$ is a separator in~$G$ so that for each component~$D_i$, there is a single tree~$T^i$ in the unrooted forest~$T - t^*$ such that~$T^i$ contains all nodes whose bags contain some~$v \in V(D_i)$; see for example~\cite[(2.3)]{RobertsonS86}.

The choice of~$t^*$ ensures that~$|V(D_i) \cap S| < 2k+3$ for all~$i \in [p]$: when vertices of~$D_i$ are contained in bags of a tree rooted at a child of~$t^*$ this follows from the fact that~$t^*$ is a deepest node for which~$|S \cap \chi(T_{t^*})| \geq 2k+3$; when vertices of~$D_i$ are contained in the tree~$T^i$ of~$T - t^*$ having the parent of~$t^*$, this follows from the fact that~$\chi(T_{t^*})$ contains at least~$2k+3$ vertices from~$S$, none of which appear in~$D_i$ since a vertex occurring in~$\chi(T_{t^*})$ and in a bag outside~$T_{t^*}$, is contained in~$\chi(t^*)$ and therefore part of the separator~$\chi(t^*)$ used to obtain the component~$D_i$. Hence none of the vertices of~$S \cap \chi(t^*)$ can appear in~$D_i$, which means there are at most~$|S|-(2k+3)\leq k+1$ vertices in~$V(D_i) \cap S$.

Since~$|S| = 3k+4$ and no component~$D_i$ contains at least~$2k+3$ vertices from~$S$, the components can be partitioned into two parts~$\mathcal{D}_1, \mathcal{D}_2$ such that~$\sum _{D_i \in \mathcal{D}_j} |V(D_i) \cap S| \leq 2k+2$ for each~$j \in \{1,2\}$. If some component contains at least~$k+2$ vertices from~$S$, then that component is a part by itself, ensuring the remainder has at most~$3k+4-(k+2)\leq 2k+2$ vertices from~$S$; if no component contains at least~$k+2$ vertices from~$S$, then any inclusion-minimal subset of components having at least~$k+2$ vertices from~$S$ has at most~$2k+2$ of them.

Define~$S'_A := \bigcup _{D_i \in \mathcal{D}_1} V(D_i) \cap S$ and~$S'_B := \bigcup _{D_i \in \mathcal{D}_2} V(D_i) \cap S$, and assume without loss of generality that~$|S'_A| \geq |S'_B|$. Note that~$|S'_A \cup S'_B| = |S \setminus \chi(t^*)| \geq 2k+3$, so that the larger side~$S'_A$ contains at least~$k+2$ vertices. To turn~$S'_A, S'_B$ into the desired partition of~$S$, it suffices to take~$S_A = S'_A$ and~$S_B = S'_B \cup (\chi(t^*) \cap S) = S \setminus S_A$. It is clear that~$|S_A| = |S'_A| \geq k+2$, while~$|S_B| = |S| - |S_A| \geq 3k+4 - (2k+2) \geq k+2$. The fact that $|S_A|,|S_B| \geq k+2$ while they partition $S$ with $|S| = 3k+4$ implies $|S_A|,|S_B| \leq 2k+2$ as desired.
Since~$\chi(t^*)$ separates~$S'_A$ from~$S'_B$, it separates~$S_A$ from~$S_B$ and we have~$\lambda_G(S_A, S_B) \leq |\chi(t^*)| = |\chi(t^*) \setminus L| \leq k+1$.
\end{proof}

We can now translate the last two lemmas into an algorithmic statement, which will be used as a subroutine in the main algorithm.
When $\hhtw(G) \le k$ and $S \subseteq V(G)$ is of size $3k+4$, then we can either split it in a balanced way, split off a base component, or detect a redundancy in a given $\hh$-deletion set and reduce its size.
Each of these outcomes will guarantee some progress for the task of constructing a tree $\hh$-decomposition.

\begin{lemma}
\label{lem:subset:subroutine}
Let $\hh$ be a hereditary and union-closed class of graphs.
There is an algorithm that,
using oracle-access to an algorithm~$\mathcal{A}$ for \hhdel,
takes as input an $n$-vertex $m$-edge graph $G$, integer $k$, $\hh$-deletion set $X$ in $G$, and a set $S \subseteq V(G)$ of size $3k+4$, 
runs in time $\Oh(8^k\cdot k(n+m))$ and polynomial space, makes  $\Oh(8^k)$ calls to $\mathcal{A}$ on induced subgraphs of $G$ and parameter $2k+2$, and terminates with one of the following outcomes.
 \begin{enumerate}
     \item\label{outcome:partition} A partition $S = S_A \cup S_B$ and a separation $(A,B)$ in $G$ are returned, such that 
     $S_A \subseteq A$, $S_B \subseteq B$, $|S_A| \le 2k+2$, $|S_B| \le 2k+2$, and $|A \cap B| \le k+1$.
     \item\label{outcome:base} A subset $S' \subseteq S$ and a separation $(A,B)$ in $G$ are returned, such that
     $S' \subseteq A$, $X \subseteq B$, $|S'| =  2k+3$, and $|A \cap B| \le 2k+2$. \bmp{(This implies that~$G[A \setminus B] \in \hh$.)}
     \item\label{outcome:redundant} An $\hh$-deletion set $X'$ in $G$ is returned, that is smaller than $X$.
     \item The algorithm correctly concludes that $\hhtw(G) > k$.
 \end{enumerate}
\end{lemma}
\begin{proof}
The algorithm starts by trying to reach the first outcome. For each partition~$S_A \cup S_B$ of~$S$ in which both parts have at most~$2k+2$ vertices, it performs at most~$k+2$ iterations of the Fold-Fulkerson algorithm to test whether~$\lambda_G(S_A, S_B) \leq k+1$. If so, then the algorithm outputs a corresponding separation~$(A,B)$ in~$G$ with~$S_A \subseteq A$,~$S_B \subseteq B$, and~$|A \cap B| = \lambda_G(S_A,S_B) \leq k+1$. By \cref{cor:subset:ford}, this can be done in time~$\Oh(k(n+m))$.

Next, the algorithm attempts to reach the second outcome. For each subset~$S' \subseteq S$ of size~$2k+3$, it performs at most~$2k+3$ iterations of the Ford-Fulkerson algorithm to test whether~$\lambda_G(S', X) \leq 2k+2$. If so, the algorithm extracts a corresponding separation~$(A,B)$ with~$S' \subseteq A$,~$X \subseteq B$, and~$|A \cap B| \leq 2k+2$, and outputs it.

If the algorithm has not terminated so far, it will reach the third or fourth outcome. It proceeds as follows.
\begin{enumerate}
    \item For each subset~$S' \subseteq S$ of size~$2k+3$, we have~$\lambda_G(S', X) > 2k+2$ since we could not reach the second outcome. As~$|S'| = 2k+3$ this implies~$\lambda_G(S',X) = 2k+3$. By Menger's theorem, there is a packing~$\mathcal{P}_{S'}$ of~$2k+3$ vertex-disjoint~$(S',X)$-paths, and such a packing can be extracted from the final stage of the Ford-Fulkerson computation.
    \item Let~$X'_{S'} \subseteq X$ be the endpoints in the set~$X$ of the paths~$\mathcal{P}_{S'}$, so that~$|X'_{S'}| = |S'| = 2k+3$. 
    \item We invoke algorithm~$\mathcal{A}$ on the graph~$G - (X \setminus X'_{S'})$ and parameter value~$2k+2$, to find a minimum-size $\hh$-deletion set in~$G - (X \setminus X'_{S'})$ or conclude that such a set has size more than~$2k+2$. If~$\mathcal{A}$ returns a solution~$Y$ of size at most $2k+2$, then~$(X \setminus X'_{S'}) \cup Y$ is an $\hh$-deletion set in~$G$ smaller than~$X$ and we return it as the third outcome.
\end{enumerate}

If none of the preceding steps caused the algorithm to give an output, then we conclude that~$\hhtw(G) > k$ and terminate.

\subparagraph*{Correctness.} We proceed to argue for correctness of the algorithm. It is clear that if the algorithm terminates with one of the first three outcomes, then its output is correct. We proceed to show that if~$\hhtw(G) \leq k$, then it will indeed terminate in one of those outcomes. So assume~$\hhtw(G) \leq k$, which means we may apply Lemma~\ref{lem:split:or:separable} to~$S$ and~$G$. If Case~\ref{case:split} of Lemma~\ref{lem:split:or:separable} holds, then the algorithm will detect the corresponding separation in the first phase of the algorithm and terminate with a suitable separation. 
So assume Case~\ref{case:separable} holds, so that there is a set~$S' \subseteq S$ of size~$2k+3$ which is weakly~$(\hh,k+1)$-separable. Since the set~$S'$ \bmp{is a candidate for reaching the second outcome, if that outcome is not reached} we have~$\lambda_G(S',X) > 2k+2$ and hence~$\lambda_G(S',X) = 2k+3 = |S'|$. 
Consider the family of $(S',X)$-paths~$\mathcal{P}_{S'}$ constructed by the algorithm for this choice of~$S'$ and let~$X'_{S'}$ be their endpoints in~$X$. The paths~$\mathcal{P}_{S'}$ show that~$\lambda_G(S',X'_{S'}) = |S'| = |X'_{S'}| = 2k+3$. Now we can apply Lemma~\ref{lem:subset:redundant} for~$\ell=k+1$ to infer that~$\jjh{X'_{S'}}$ is redundant in~$X$, which implies that~$G - (X \setminus X'_{S'})$ has an $\hh$-deletion set smaller than~$|X'_{S'}| = 2k+3$. Hence algorithm~$\mathcal{A}$ outputs an $\hh$-deletion set smaller than~$|X'_{S'}|$ and the algorithm terminates with the third outcome. 

Since the algorithm reaches one the first three outcomes when~$\hhtw(G) \leq k$, the algorithm is correct when it reaches the last outcome.

\subparagraph*{Running time and oracle calls.} Each of the three phases of the algorithm consist of enumerating subsets~$S' \subseteq S$, of which there are~$2^{|S|} \leq 2^{3k+4} = \Oh(8^k)$. For each such set~$S'$, the algorithm performs~$\Oh(k)$ rounds of the Ford-Fulkerson algorithm in time~$\Oh(k(n+m))$. In the last phase, the algorithm additionally invokes~$\mathcal{A}$ on an induced subgraph of~$G$ for each~$S'$ to find an $\hh$-deletion set of size at most~$2k+2$ if one exists. It follows that the running time of the algorithm (not accounting for the time spent by~$\mathcal{A}$) is~$\Oh(8^k \cdot k(n+m))$. The space usage is easily seen to be polynomial in the input size since the algorithm is iterative. This concludes the proof of Lemma~\ref{lem:subset:subroutine}.
\end{proof}

\subsection{The decomposition algorithm}

We retrace the proof of \cite[Theorem 7.18]{CyganFKLMPPS15}
which gives the classic algorithm for approximating (standard) treewidth.
{Consider sets $S \subseteq W \subseteq V(G)$ such that $\partial_G(W) \subseteq S$ and $|S| = 3k+4$; we aim to construct a tree decomposition of $G[W]$ which contains $S$ in its root bag.
We can consider all ways to partition $S$ into $S_A \cup S_B$ such that~$|S_A|, |S_B| \leq 2k+2$ 
and compute a minimum $(S_A,S_B)$-separator. Since~$|S| = 3k+4$, there are $2^{3k+4} = \Oh(8^k)$ such partitions. 
When $\tw(G) \le k$, we are guaranteed that for some partition  $S = S_A \cup S_B$ we will 
find a separator in $G[W]$ of size $\le k+1$ which yields the \bmp{separation} $(A_W,B_W)$ \mic{ in $G[W]$} satisfying $S_A \subseteq A_W$, $S_B \subseteq B_W$, and $|A_W\cap B_W| \le k+1$.
Then the boundary $\partial_G(A_W)$ is contained in $S_A \cup (A_W \cap B_W)$, and similarly $\partial_G(B_W) \subseteq S_B \cup (A_W \cap B_W)$.
We create instances $(A_W, S_A \cup  (A_W \cap B_W))$ and $(B_W, S_B \cup (A_W \cap B_W))$
to be solved recursively, analogously as $(W,S)$.
Note that each of the sets $S_A \cup  (A_W \cap B_W)$, $S_B \cup  (A_W \cap B_W)$ has less \bmp{than} $3k+4$ vertices,
\mic{so we can augment each of them with one more vertex before making the recursive call while preserving the size invariant.
This step ensures that the recursion tree has at most $|V(G)|$ nodes.} 
After computing tree decompositions for $G[A]$ and $G[B]$ we merge them by creating a new root with a bag $S \cup  (A_W \cap B_W)$ of size at most $4k + 5$.
Hence, we are able to construct a tree decomposition of width $4k+4$ assuming that one of width $k$ exists.}

There are two differences between the outlined algorithm and ours, while the recursive scheme stays the same.
First, due to scenario \eqref{outcome:base} in \cref{lem:subset:subroutine} we need to handle the cases where we can directly create a base component containing at least $2k+3$ vertices from $S$. 
The lower bound $2k+3$ is greater than the separator size $2k+2$ so we will move on to a subproblem where $S$ is significantly smaller.
\mic{We need to include the separator of size $2k+2$ in the root bag, together with $S$, so we obtain a slightly weaker bound on the maximum bag size, that is $5k + 6$.}
Next, due to scenario \eqref{outcome:redundant} we might not make direct progress in the recursive scheme but instead we reduce the size of an $\hh$-deletion set~$X$ that we maintain \mic{(which initially contains all vertices).}
This situation can happen at most $|V(G)|$ many times, so eventually we will reach outcome \eqref{outcome:partition} or \eqref{outcome:base}.

We introduce the following operation which will come in useful for merging decompositions produced by solving two subproblems recursively.

\begin{definition}\label{def:subset:merge}
Let $A,B$ be vertex sets in a graph $G$
and $(T_A, \chi_A, L_A)$, $(T_B, \chi_B, L_B)$ be tree $\hh$-decompositions of $G[A]$, $G[B]$, respectively, with roots $r_A, r_B$.
We define the \emph{merge} $(T,\chi,L)$ of $(T_A, \chi_A, L_A)$ and $(T_B, \chi_B, L_B)$ \mic{along a given set $R \subseteq V(G)$}. 
We construct $T$ by taking a disjoint union of $T_A,T_B$ and inserting a root node $r$ with children $r_A, r_B$.
We define $L = L_A \cup L_B$ and $\chi$ as follows: 
     $\chi_{\vert {T_A}} = \chi_A$, $\chi_{\vert {T_B}} = \chi_B$ and, $\chi(r) = R$. 
\end{definition}

The merge of two tree $\hh$-decompositions along $R$ may not be a valid tree $\hh$-decomposition. Whenever we apply the concept of merge, we shall prove that the merge is valid. We restate \cref{thm:intro:treewidth} for \bmp{readability}.

\restTreewidthAlgorithm*

\begin{proof}
{We shall provide an algorithm that recursively solves the following subproblem. 
Our final goal is to solve
\textsc{Decompose$(G,k,\emptyset,V(G))$}.} 

\defproblem{Decompose$(G,k,S,W)$}
{Graph $G$, integer $k$, sets $S \subseteq W \subseteq V(G)$, such that $\partial_G(W) \subseteq S$ and $|S| \le 3k + 3$.}
{Construct a tree $\hh$-decomposition $(T, \chi, L)$ of $G[W]$ of width at most $5k + 5$ such that $S \cap L = \emptyset$ and $S$ is contained in the root bag of $(T,\chi,L)$, or correctly report that $\hhtw(G) > k$.}

We will maintain an $\hh$-deletion set $X$ in $G$ as a `global variable'.
Initially we set $X = V(G)$.
Given a subproblem \textsc{Decompose}$(G,k,S,W)$ we shall either solve it directly, or reduce it to at most two smaller subproblems (measured by the size of $W$), or decrease the size of $X$ and make another attempt to solve the same subproblem with the smaller deletion set.
The last scenario cannot happen more than $n$ times, therefore finally we will be able to make progress in the recursion.
We intentionally do not pass $X$ as an argument in the recursion for the sake of optimizing the dependency on $n$ in the running time.
Treating $X$ as a global variable allows us to upper bound the total number of times when $X$ is refined during the entire computation.

Consider a subproblem \textsc{Decompose}$(G,k,S,W)$ and an $\hh$-deletion set $X$ in $G$. First, if $|W| \le 5k + 6$ we can simply return a tree $\hh$-decomposition consisting of a single node with a bag $W$ and having $L = \emptyset$. 
Assume from now on that $|W| > 5k + 6$. 
This, in particular, allows us to choose a set $\widehat{S} \subseteq W$ of size exactly $3k+4$ such that ${S} \subsetneq \widehat S$ (the choice of $\widehat{S} \setminus S$ is arbitrary; this step is important only for the running time analysis).
We execute the algorithm from \cref{lem:subset:subroutine} for $(G,k,X,\widehat S)$ 
and proceed \bmp{according} to the outcome received.

 \begin{enumerate}
     \item \bmp{Suppose} we obtain a partition $\widehat S = S_A \cup S_B$ and a separation $(A,B)$ in $G$, such that 
     $S_A \subseteq A$, $S_B \subseteq B$, $|S_A| \le 2k+2$, $|S_B| \le 2k+2$, and $|A \cap B| \le k+1$.
    Then $(A_W, B_W) = (A \cap W, B \cap W)$ is a separation in $G[W]$.
    Next, we set $\widehat{S}_A = S_A \cup (A_W \cap B_W)$,  $\widehat{S}_B = S_B \cup (A_W \cap B_W)$, and create instances $(G,k,\widehat{S}_A,A_W)$, $(G,k,\widehat{S}_B,B_W)$ to be solved recursively.
    Note that $\partial_G(A_W) \subseteq \widehat{S}_A$,  $\partial_G(B_W) \subseteq \widehat{S}_B$, and $|\widehat{S}_A|,  |\widehat{S}_B| \le (2k+2) + (k+1) = 3k+3$ therefore these instances satisfy the preconditions of the problem \textsc{Decompose}.

    \mic{If for any of the subproblems we obtain \bmp{the} conclusion that $\hhtw(G) > k$, we \mic{report it as the outcome of the current call}.
    Otherwise,} let $(T_A, \chi_A, L_A)$ and $(T_B, \chi_B, L_B)$ be tree $\hh$-\bmp{decompositions} obtained after solving instances $(G,k,\widehat{S}_A,A_W)$, $(G,k,\widehat{S}_B,B_W)$, respectively.
    We return the merge of $(T_A, \chi_A, L_A)$ and $(T_B, \chi_B, L_B)$ \mic{along $\hat{S} \cup (A_W \cap B_W)$}.
    (Recall \cref{def:subset:merge}). 
    
     
     \item \bmp{Suppose} we obtain a subset $S' \subseteq \hat{S}$ and a separation $(A,B)$ in $G$, such that
     $S' \subseteq A$, $X \subseteq B$, $|S'| =  2k+3$, and $|A \cap B| \le 2k+2$.
    Again, $(A_W, B_W) = (A \cap W, B \cap W)$ is a separation in $G[W]$.
     We set $\widehat{S}_A = (\widehat S \cap A_W) \cup (A_W \cap B_W)$,  $\widehat{S}_B =  (\widehat S \sm A_W) \cup (A_W \cap B_W)$.
     We have $\partial_G(A_W) \subseteq \widehat{S}_A$ and $\partial_G(B_W) \subseteq \widehat{S}_B$.
     The size of $\widehat{S}_A$ is at most $|\widehat{S}| + |A_W \cap B_W| \le (3k+4) + (2k+2) = 5k + 6$ while the size of 
      $\widehat{S}_B$ is at most $|\widehat{S} \sm S'| + |A_W \cap B_W| \le (k+1) + (2k+2) = 3k + 3$.
    Then  $(G,k,\widehat{S}_B,B_W)$ is a valid instance of \textsc{Decompose}.
    The size of  $\widehat{S}_A$ does not satisfy the invariant though and we handle the instance $(G,k,\widehat{S}_A,A_W)$ differently.
    
    As $X \subseteq B$ and $A_W \cap B_W \subseteq \widehat{S}_A$ we get $(A_W \sm \widehat{S}_A) \cap X = \emptyset$.
    Because $X$ is an $\hh$-deletion set, this means that \bmp{$G[A_W \sm \widehat{S}_A] \in \hh$ as it is an induced subgraph of~$G-X \in \hh$ while~$\hh$ is hereditary}.
    We construct a tree $\hh$-decomposition  $(T_A, \chi_A, L_A)$ of $G[A_W]$ as follows:
    we create a root node with a bag $\widehat{S}_A$ having a single child with a bag $A_W$
    and we set $L_A = A_W \sm \widehat{S}_A$. 
    The width of this decomposition equals $|\widehat{S}_A| - 1 \le 5k + 5$.
    Note that $(T_A, \chi_A, L_A)$ satisfies the output specification of \textsc{Decompose} for $(G,k,\widehat{S}_A,A_W)$.
    Finally, we solve \textsc{Decompose}$(G,k,\widehat{S}_B,B_W)$.
    \mic{If this call reports that $\hhtw(G) >k$, we propagate this outcome.}
    Otherwise, we return the merge of the two computed tree $\hh$-decompositions \mic{along $\hat{S} \cup (A_W \cap B_W)$}. 
      
     \item \bmp{Suppose} an $\hh$-deletion set $X'$ in $G$ is returned, that is smaller than $X$.
     In this case we update $X \leftarrow X'$ and repeat the process, applying \cref{lem:subset:subroutine} again. 
     
     \item If the subroutine reports that $\hhtw(G) > k$ then we \mic{return this outcome for the current~call.}
 \end{enumerate}

\subparagraph*{Correctness.}
We argue that the construction given in cases (1) and (2) yields a valid output for \textsc{Decompose}.
In both cases we deal with a superset $\widehat S \subseteq W$ of $S$ and a separation $(A_W,B_W)$ of $G[W]$.
The sets $\widehat S_A, \widehat S_B$ satisfy $\partial_G(A_W) \subseteq \widehat S_A$, $\partial_G(B_W) \subseteq \widehat S_B$, and $\widehat S \subseteq \widehat S_A \cup \widehat S_B$.

Let $(T_A, \chi_A, L_A)$ and $(T_B, \chi_B, L_B)$ be tree $\hh$-decompositions obtained after solving instances $(G,k,\widehat{S}_A,A_W)$, $(G,k,\widehat{S}_B,B_W)$, either by recursion or by the construction for a leaf node in case (2).
We refer to the roots of $T_A,T_B$ as $r_A, r_B$, respectively.
By the output specification of \textsc{Decompose} it holds that
$\widehat S_A \subseteq \chi_A(r_A)$, and
$\widehat S_B \subseteq \chi_B(r_A)$.
We also have that $L_A \subseteq A_W \sm \widehat S_A$,  $L_B \subseteq B_W \sm \widehat S_B$.

Let $(T,\chi,L)$ be the merge of $(T_A, \chi_A, L_A)$ and $(T_B, \chi_B, L_B)$ \mic{along $\hat{S} \cup (A_W \cap B_W) = \widehat S_A \cup \widehat S_B$}.
We have $|\chi(r)| 
\le |\widehat S| + |A_W \cap B_W| \le (3k+4) + (2k + 2) = 5k + 6$ so we keep the maximum bag size in check. 
Next, $S$ is contained in $\chi(r)$ and is disjoint from both $L_A, L_B$.
It remains to show that $(T,\chi,L)$ is a valid tree $\hh$-decomposition of $G[W]$.
We check the conditions of \cref{def:tree:h:decomp}.

Let $v \in W$: we show that $Y_v = \{t \in V(T) \mid t \in \chi(t)\}$ is non-empty and connected.
If $v \in A_W \setminus \widehat{S}_A$, then $v\not\in B_W$ and $Y_v$ is a non-empty connected subtree of $T_A$. 
If $v \in \widehat{S}_A \setminus B_W$, then $Y_v \cap V(T_A)$ is connected and contains $r_A$, so adding $r$ to $Y_v \cap V(T_A)$ does not affect connectivity.
After considering the symmetric cases, it remains to check $v \in \widehat{S}_A \cap \widehat{S}_B$. 
The set $Y_v$ is then a union of $\{r\}$ and two connected subtrees, each containing a child of $r$, so it is connected in $T$.

Now consider an edge $uv \in E(G[W])$.
Since $(A_W,B_W)$ is a {separation} of $G[W]$, \bmp{we have} $\{u,v\} \subseteq A_W$ or $\{u,v\} \subseteq B_W$ (or both).
Hence, there exists a node $t \in V(T_A) \cup V(T_B)$ so that  $\{u,v\} \subseteq \chi(t)$.

Next, let $v \in L_A$.
There is a unique leaf node $t_v \in V(T_A)$ for which $v \in \chi(t_v)$.
By the output specification, $L_A \subseteq A_W \setminus \widehat{S}_A$.
Hence, $v \not\in \chi(r)$ and $v \not\in \chi(t)$ for any $t \in V(T_B)$, so $t_v$ is the unique node in $T$ whose bag contains $v$; it remains a leaf after insertion of $r$.
The case $v \in L_B$ is analogous.

Finally, any leaf $t \in V(T)$ is either a leaf in $T_A$ or in $T_B$.
We get $G[\chi(t) \cap L] \in \hh$ as a~direct consequence of this property for  $(T_A, \chi_A, L_A)$ and $(T_B, \chi_B, L_B)$.

\subparagraph*{Running time and oracle calls.}
We show that the number of processed instances of \textsc{Decompose} is $\Oh(n)$.
We say that an instance is basic if it does not create other instances recursively.
For an instance $(G,k,S,W)$ we argue that the total number of non-basic nodes in its {recursion} tree is at most $|W \setminus S|$.
This holds trivially if the instance is basic.
In case (1)
 we {recurse} into two instances $(G,k,\widehat{S}_A,A_W)$, $(G,k,B,\widehat{S}_B, B_W)$ 
for some {separation} $(A_W,B_W)$ in $G[W]$
such that  $\widehat{S} \cap A_W \subseteq \widehat{S}_A$ and $\widehat{S} \cap B_W \subseteq \widehat{S}_B$. 
Then $(A_W \setminus \widehat{S}_A, B_W \setminus \widehat{S}_B)$ is a partition of a proper subset of $W \setminus S$ because $\widehat{S} \supsetneq S$.
This implies that $|A_W \setminus \widehat{S}_A|  + |B_W \setminus \widehat{S}_B| + 1 \le |W \setminus S|$.
In case (2) we also create two instances satisfying the specification above, but only the second one is being solved recursively and the first one is solved directly.
Therefore, the inequality above holds also in this case.
The claim follows by an induction on the depth of the recursion tree.
Finally, either the root instance is basic or the number of basic nodes is at most twice the number of non-basic ones.
As a consequence, the total number of nodes in the recursion tree is $\Oh(n)$.
This implies that the number of nodes in the computed decomposition is $\Oh(n)$.

A single execution of the algorithm from \cref{lem:subset:subroutine} takes time $\Oh(8^k \cdot k(n+m))$
and $\Oh(8^k)$ calls to the oracle $\mathcal{A}$.
The number of executions is bounded by the number of processed instances of \textsc{Decompose} plus the number of times we have refined $X$.
Since the size of $X$ can drop at most $n$ times, we infer that the number of executions is $\Oh(n)$.  
This concludes the proof of the theorem.
\end{proof}

\section{Applications}
\label{sec:applications}

We list several important corollaries from \cref{thm:intro:treewidth}.
For classes $\hh \in \{\mathsf{bipartite, interval}\}$ and any class $\hh$ defined by a finite family of forbidden induced subgraphs, we obtain single-exponential 5-approximations for computing $\hh$-treewidth.
For $\hh \in \{\mathsf{chordal, planar}\}$ the running-time dependency on $k = \hhtw(G)$ becomes $2^{\Oh(k\log k)}$.
For $\hh \in \{\mathsf{interval, planar}\}$ the dependency on the graph size is quadratic.

\begin{corollary}\label{cor:htw:apx}
Each of the following graph classes $\hh$ admits an algorithm
that, given an $n$-vertex $m$-edge graph $G$ and integer $k$, runs in time $f_{\hh}(k,n,m)$, and either computes a tree $\hh$-decomposition of $G$ of width at most $5k+5$ or correctly concludes that $\hhtw(G) > k$.
The function $f_{\hh}(n,m,k)$ is specified as follows.\bmpr{Do we want to include $\mathcal{H} = K_t$-free? We should be able to get single-exponential which is ETH-tight.}
\begin{itemize}
    \item $\hh = \mathsf{bipartite}: \Oh(72^k\cdot  n^2(n+m))$,
    \item $\hh = \mathsf{interval}:  \Oh(8^{3k}\cdot 
 n(n+m))$,
    \item $\hh = \mathsf{planar}: 2^{\Oh(k\log k)}\cdot  n(n+m)$,
    \item $\hh = \mathsf{chordal}:  2^{\Oh(k\log k)}\cdot n^{\Oh(1)}$,
    \item any $\hh$ defined by a finite family of forbidden induced subgraphs on at most $c$ vertices:\\ $f_{\hh}(n,m,k) = (8c^2)^k\cdot n^{\Oh(1)}$.
\end{itemize}
\end{corollary}
\begin{proof}
Let $g_{\hh}(n,m,s)$ describe the running time for solving \hhdel parameterized by the solution size $s$.
\cref{thm:intro:treewidth} yields an algorithm for approximating 
$\hh$-treewidth with running time 
$\Oh(8^k  kn(n+m)) + \Oh(8^kn) \cdot g_{\hh}(n,m,2k+2)$.
We check the state-of-the-art running times for \hhdel.
\begin{itemize}
    \item $\hh = \mathsf{bipartite}: \Oh(3^s\cdot sn(n+m))$~\cite{ReedSV04} (cf.~\cite[Thm. 4.17]{CyganFKLMPPS15}), 
    \item $\hh = \mathsf{interval}: \Oh(8^s\cdot(n+m))$~\cite{CaoK16},
    \item $\hh = \mathsf{planar}: 2^{\Oh(s\log s)}\cdot n$~\cite{DBLP:conf/soda/JansenLS14},
    \item $\hh = \mathsf{chordal}: 2^{\Oh(s\log s)}\cdot n^{\Oh(1)}$~\cite{CaoM16}.
\end{itemize}
The algorithm for $\hh = \mathsf{chordal}$ is presented for the decision version \hhdel, but it can easily be transformed into an algorithm which constructs a minimum solution if there exists one of size at most~$k$ by a self-reduction (we do not care about the dependency on $n$ in this case).
When $\hh$ is defined by a finite family of forbidden induced subgraphs on at most $c$ vertices, \hhdel is solvable in time $c^s\cdot n^{\Oh(1)}$~\cite{Cai96}. 
It suffices to plug these running times into \cref{thm:intro:treewidth}.
\end{proof}

Next, we derive faster algorithms for \hhdel under the parameterization by \hhtwfull{}.
The obtained running times for $\hh \in \{\mathsf{bipartite, planar}\}$ match the running times under  the parameterization by treewidth, which are known to be ETH-tight.

\begin{corollary}\label{cor:htw:hhdel}
The following graph classes $\hh$ admit algorithms for \hhdel parameterized by $k = \hhtw(G)$ with the running time $f_{\hh}(k) \cdot n^{\Oh(1)}$, where the function $f_{\hh}$ is specified as follows.
\begin{itemize}
    \item $\hh = \mathsf{bipartite}: 2^{\Oh(k)}$,
    \item $\hh = \mathsf{planar}: 2^{\Oh(k\log k)}$,
    \item $\hh = \mathsf{chordal}:  2^{\Oh(k^2)}$,
    \item any $\hh$ defined by a finite family of forbidden induced subgraphs on at most $c$ vertices:\\ $f_{\hh}(k) = 2^{\Oh(k^{2c})}$.
\end{itemize}
\end{corollary}
\begin{proof}
By \cref{cor:htw:apx}, we can construct a tree $\hh$-decomposition of width $\Oh(k)$ within the claimed running time.
Each considered \hhdel{} problem can be solved in time $f_{\hh}(d) \cdot n^{\Oh(1)}$, where $d$ is
the width of a~given tree $\hh$-decomposition~\cite{JansenKW21-arxiv}.
\end{proof}

We have not considered $\hh = \mathsf{interval}$ here because the known algorithm working on a given tree $\hh$-decomposition~\cite{JansenKW21-arxiv} runs in time  $2^{\Oh(d^c)} \cdot n^{\Oh(1)}$ for a large constant $c$, which gives a prohibitively high running time even combined with our 5-approximation for $\hh$-treewidth.

\section{\texorpdfstring{Approximating $\hh$-elimination distance}{Approximating H-elimination distance}} \label{sec:treedepth}

We switch our attention to a different kind of a hybrid graph measure, namely $\hh$-elimination distance. 
We only provide the definition of the measure itself since we do not work with the corresponding decompositions directly. See \cite{JansenKW21} for more details.

\begin{definition}
For a hereditary graph class $\hh$ and a graph $G$, the \hhdepthfull{} of $G$, denoted $\hhdepth(G)$, is defined recursively as follows.
\begin{equation*}
    \hhdepth(G) = \begin{cases}
        0 & \mbox{if $G$ is connected and $G \in \hh$} \\
        1 + \min_{v \in V(G)}(\hhdepth(G-v)) & \mbox{if $G$ is connected and $G \not\in \hh$} \\
         \max_{i=1}^d \hhdepth(G_i) & \mbox{if $G$ is disconnected and $G_1, \dots G_d$ are its components}
    \end{cases}
\end{equation*}
A tree structure which encodes this recursion (not necessarily of optimal depth) is called an $\hh$-elimination forest of $G$.
\bmp{For $\hh$ consisting of only the 0-vertex graph, \bmp{$\hhdepth(G)$} is the treedepth of~$G$ and the corresponding structure is a (standard) elimination forest.}
\end{definition}

We exploit the concept of redundancy to improve the bottleneck of the existing approximation algorithms for computing a decomposition of small depth~\cite{JansenKW21}.
This bottleneck involves \mic{repeatedly} finding an $(\hh,k)$-separation that weakly covers a given vertex set $Z$.
Such a subroutine is used to detect subgraphs which potentially may be turned into base components.
As this approach is not aimed at constructing an $\hh$-elimination forest of an~optimal depth, one can relax the constraint on the neighborhood size and seek an $(\hh, k')$-separation, \mic{where $k'$ is upper bounded in terms of $k$.} 

We define the \textsc{$\hh$-Weak Coverage} problem, where the input consists of a graph $G$, a non-empty set $Z \subseteq V(G)$, and an integer $k$.
An $\alpha$-approximate algorithm for \textsc{$\hh$-Weak Coverage} should either return an $(\hh,\alpha k)$-separation weakly covering $Z$ or conclude that there is no $(\hh,k)$-separation weakly covering $Z$. 

Jansen, de Kroon, and W\l odarczyk~\cite{JansenKW21} showed that an FPT approximation algorithm for \textsc{$\hh$-Weak Coverage} implies an FPT algorithm for constructing an $\hh$-elimination forest of approximately optimal width.
In fact, they worked with a version of \textsc{$\hh$-Weak Coverage} with stronger assumptions: that the graph $G$ should have bounded $\hh$-treewidth, that $G[Z]$ should be connected, and the algorithm could report a failure \mic{already when} there is no `{strong}' coverage for $Z$, i.e., there is no $(\hh,k)$ separation $(C,S)$ with $Z \subseteq C$
(whereas in weak coverage we have $Z \subseteq C \cup S$).
Remarkably, we do not need these assumptions.
We reformulate the original lemma to be consistent with the definition of the \textsc{$\hh$-Weak Coverage} problem and our oracle-based formalism.

\begin{lemma}[{\cite[Lem. 3.3]{JansenKW21}}]
\label{lem:separation-finding:motivation}
Let $\hh$ be a hereditary and union-closed class of graphs.
There is an algorithm that, using oracle-access to an $\Oh(1)$-approximate algorithm~$\mathcal{B}$ for \textsc{$\hh$-Weak Coverage}, takes as input an $n$-vertex graph $G$ with $\hh$-elimination distance at most~$k$,
runs in time $n^{\Oh(1)}$, 
makes $n^{\Oh(1)}$ calls to $\mathcal{B}$ on graph $G$ and parameter $k$,
and returns an $\hh$-elimination forest of $G$ of depth\footnote{We remark that the exponent $\frac 3 2$ at $(\log k)$ is missing in the statement from the conference version of the article~\cite{JansenKW21}.
This factor comes directly from the best-known polynomial-time approximation algorithm for treedepth~\cite{CzerwinskiNP19}.} $\Oh(k^3\log^{3 / 2} k)$.

{Under the same assumptions, there is an algorithm that
runs in time $2^{\Oh(k^2)} \cdot n^{\Oh(1)}$,
makes $n^{\Oh(1)}$ calls to $\mathcal{B}$ on graph $G$ and parameter $k$,
and returns an $\hh$-elimination forest of $G$ of depth $\Oh(k^2)$.}
\end{lemma}

The difference between the two statements is caused by the usage of different algorithms for constructing an elimination forest (either approximate or exact) which occurs after a preliminary decomposition of the graph is constructed with the help of the algorithm $\mathcal{B}$.
The advantage of the first algorithm is that,
besides the calls to $\mathcal{B}$, its running time is polynomial.
As a consequence, it requires only polynomial space as long as  $\mathcal{B}$ runs in polynomial space.

A 2-approximation algorithm for \textsc{$\hh$-Weak Coverage} already follows from the proof of \cref{lem:subset:subroutine} but we present it below in a stand-alone form.
In the previous work, the linear dependence on $k$ was obtained only in two special cases \bmp{(bipartite graphs, or graphs defined by a finite family of forbidden induced subgraphs)} while
for the remaining graph classes for which \hhdel is FPT by the solution size, an algorithm was given that returns an $(\hh,\Oh(k^2))$-separation.

\begin{lemma}\label{lem:separation-finding:algorithm}
Let $\hh$ be a hereditary and union-closed class of graphs. 
Assuming oracle-access to an algorithm~$\mathcal{A}$ for \hhdel, \textsc{$\hh$-Weak Coverage} admits a 2-approximate algorithm that, given an $n$-vertex $m$-edge graph $G$ and integer $k$, runs in time $\Oh(kn(n+m))$, and
makes at most $n$ calls to $\mathcal{A}$ on induced subgraphs of $G$ and parameter $2k$.
\end{lemma}
\begin{proof}
{Consider an input~$(G,Z,k)$ to the \textsc{$\hh$-Weak Coverage} problem.} 
During the algorithm, we maintain an $\hh$-deletion set $X$, initialized as $X = V(G)$.
We repeatedly perform the following.
First, we apply \cref{cor:subset:ford} to check whether $\lambda_G(Z,X) \le 2k$ in time $\Oh(k(n+m))$.
If yes, we obtain a $(Z,X)$-separator $S$.
The set $C$, comprising vertices reachable from $Z \sm S$ in $G-S$, is disjoint from $X$ hence $G[C]$ is an induced subgraph of $G-X$ and so it belongs to $\hh$.
Therefore $(C,S)$ forms a $(\hh,2k)$-separation.
We also have $Z \subseteq C \cup S$ so $(C,S)$ weakly covers $Z$ and we can return it as a solution.

Now suppose that $\lambda_G(Z,X) \ge 2k + 1$.
Then by \cref{cor:subset:ford} we obtain
a family of $2k+1$ vertex-disjoint $(Z,X)$-paths.
Let $X' \subseteq X$ be the set of $X$-endpoints of these paths.
Then $|X'| = 2k+1$ and $\lambda_G(Z,X') = 2k + 1$.
We execute the algorithm $\mathcal{A}$ on the graph $G - (X \setminus X')$ and parameter $2k$.
When $Z$ is weakly $(\hh,k)$-separable then, by \cref{lem:subset:redundant}, the set $X'$ is redundant in $X$.
In this case the algorithm $\mathcal{A}$ will find a set $X'' \subseteq V(G)$ of size at most $2k = |X'| - 1$ such that $(X \setminus X') \cup X''$ is also an $\hh$-deletion set in $G$.
We set $X \leftarrow (X \setminus X') \cup X''$ and continue this process.
Observe that in each iteration the size of $X$ decreases so after at most $n$ steps we either arrive at the scenario $\lambda_G(Z,X) \le 2k$ (then we find a solution) or $\lambda_G(Z,X) \ge 2k + 1$ but the call to  $\mathcal{A}$ fails to find a local improvement.
\cref{lem:subset:redundant} implies that then \bmp{there is no $(\hh,k)$-separation weakly covering~$Z$} and so we can report a failure.
\end{proof}

Combining \cref{lem:separation-finding:motivation} with \cref{lem:separation-finding:algorithm} yields \cref{thm:intro:elimination-distance}.

\section{Conclusion} \label{sec:conclusion}

We contributed to the algorithmic theory of hybrid graph parameterizations, by showing how a 5-approximation to $\hhtw$ can be obtained using an algorithm for the solution-size parameterization of \hhdel as a black box. This makes the step of computing a tree $\hh$-decomposition now essentially as fast as that of solving \hhdel parameterized by solution size. Our new decomposition algorithm combines with existing algorithms to solve \hhdel{} on a given tree $\hh$-decomposition, to deliver algorithms that solve \hhdel{} parameterized by $\hhtw$. For \textsc{Odd Cycle Transversal} and \textsc{Vertex Planarization}, the parameter dependence of the resulting algorithm is equal to the worst of the parameter dependencies of the solution-size and treewidth-parameterizations. We believe that this is not a coincidence, and offer the following conjecture.

\begin{conjecture}
Let~$\hh$ be a hereditary and union-closed graph class. \bmp{If \hhdel{} \mic{can} be solved in time~$f(s) \cdot n^{\Oh(1)}$ parameterized by solution size~$s$, and in time~$h(w) \cdot n^{\Oh(1)}$ parameterized by treewidth~$w$, then \hhdel{} can be solved in time~$\left (f(\Oh(k)) + h(\Oh(k) \right)) \cdot n^{\Oh(1)}$ parameterized by $\hh$-treewidth~$k$.}
\end{conjecture}

The conjecture is a significant strengthening of the equivalence, with respect to non-uniform fixed-parameter tractability, between solving $\hhdel{}$ parameterized by solution size and computing $\hhtw$ given by Agrawal et al.~\cite{AgrawalKLPRSZ22}. It essentially states that there is no \emph{price of generality} to pay for using the hybrid parameterization by $\hhtw$. After three decades in which the field of parameterized complexity has focused on parameterizations by solution size, this would lead to a substantial shift of perspective. We believe Theorem~\ref{thm:intro:treewidth} is an important ingredient in this direction. 

To understand the relative power of the parameterizations by solution size, treewidth, and $\hh$-treewidth, the remaining bottleneck lies in using the tree $\hh$-decomposition to compute a minimum $\hh$-deletion set. Can the latter be done as efficiently when using a tree $\hh$-decomposition as when using a standard tree decomposition? For problems like \textsc{Odd Cycle Transversal} and \textsc{Vertex Planarization}, this is indeed the case. But when the current-best dynamic-programming algorithm over a tree decomposition uses advanced techniques, it is currently not clear how to lift such an algorithm to work on a tree $\mathcal{H}$-decomposition. Can \hhdel{} for $\hh$ the class of interval graphs be solved in time~$2^{\Oh(k \log k)} \cdot n^{\Oh(1)}$ parameterized by~$\hhtw$? Such a running time can be obtained for the parameterization by treewidth by adapting the approach of Saitoh, Yoshinaka, and Bodlaender~\cite{SaitohYB21}.

While we have not touched on the subject here, we expect our ideas to also be applicable when~$\hh$ is a \emph{scattered graph class}, i.e., when~$\hh$ 
consists of graphs where each connected component is contained in one of a finite number of graph classes~$\mathcal{H}_1, \ldots, \mathcal{H}_t$. It is known~\cite{JansenKW21} that, when \textsc{Vertex Cover} can be solved in polynomial time on each graph class~$\mathcal{H}_i$, then \textsc{Vertex Cover} is FPT parameterized by the width of a given tree $\mathcal{H}$-decomposition. We expect that \cref{thm:intro:treewidth} can be generalized to work with scattered graph classes $\hh$, as long as there is an oracle to solve \textsc{$\mathcal{H}_i$-deletion} parameterized by solution size for each individual class~$\mathcal{H}_i$. To accommodate this setting, the algorithm maintains an $\mathcal{H}_i$-deletion set~$X_i$ for \emph{each} graph class~$\mathcal{H}_i$. A step of the decomposition algorithm then either consists of finding a balanced separation of~$S$, splitting of a base component, or improving \emph{one of the} deletion sets~$X_i$ (which can occur only~$t \cdot |V(G)|$ times).

The decomposition algorithm we presented has an approximation factor of 5.
\mic{It may be possible to obtain a smaller approximation ratio at the expense of a worse base of the exponent, by \bmp{repeatedly} splitting large bags~\cite{BellenbaumD02,Korhonen21,KorhonenL22}.}
For obtaining single-exponential \hhdel algorithms, the advantage of the improved approximation factor would be immediately lost due to the increased running time and therefore we did not pursue this direction.
A final direction for future work concerns the optimization of the polynomial part of the running time. For standard treewidth, a 2-approximation can be computed in time~$2^{\Oh(k)} \cdot n$~\cite{Korhonen21}, which was obtained after a long series of improvements (cf.~\cite[Table 1]{BodlaenderDDFLP16}) on both the approximation factor and dependence on~$n$. Can a constant-factor approximation to~$\mathcal{H}$-treewidth be computed in time~$2^{\Oh(k)}\cdot (n + m)$ for graph classes~$\mathcal{H}$ like bipartite graphs?

\bibliography{bib}

\begin{thebibliography}{10}

\bibitem{AgrawalKLPRSZ22}
Akanksha Agrawal, Lawqueen Kanesh, Daniel Lokshtanov, Fahad Panolan, M.~S.
  Ramanujan, Saket Saurabh, and Meirav Zehavi.
\newblock Deleting, eliminating and decomposing to hereditary classes are all
  {FPT}-equivalent.
\newblock In Joseph~(Seffi) Naor and Niv Buchbinder, editors, {\em Proceedings
  of the 2022 {ACM-SIAM} Symposium on Discrete Algorithms, {SODA} 2022, Virtual
  Conference / Alexandria, VA, USA, January 9 - 12, 2022}, pages 1976--2004.
  {SIAM}, 2022.
\newblock \href {https://doi.org/10.1137/1.9781611977073.79}
  {\path{doi:10.1137/1.9781611977073.79}}.

\bibitem{AgrawalR22}
Akanksha Agrawal and M.~S. Ramanujan.
\newblock Distance from triviality 2.0: {H}ybrid parameterizations.
\newblock In Cristina Bazgan and Henning Fernau, editors, {\em Combinatorial
  Algorithms - 33rd International Workshop, {IWOCA} 2022, Trier, Germany, June
  7-9, 2022, Proceedings}, volume 13270 of {\em Lecture Notes in Computer
  Science}, pages 3--20. Springer, 2022.
\newblock \href {https://doi.org/10.1007/978-3-031-06678-8_1}
  {\path{doi:10.1007/978-3-031-06678-8_1}}.

\bibitem{ArnborgCP87}
Stefan Arnborg, Derek~G. Corneil, and Andrzej Proskurowski.
\newblock Complexity of finding embeddings in a k-tree.
\newblock {\em SIAM J. Algebraic Discrete Methods}, 8(2):277–284, 1987.
\newblock \href {https://doi.org/10.1137/0608024} {\path{doi:10.1137/0608024}}.

\bibitem{BelbasiF22}
Mahdi Belbasi and Martin F{\"{u}}rer.
\newblock An improvement of reed's treewidth approximation.
\newblock {\em J. Graph Algorithms Appl.}, 26(2):257--282, 2022.
\newblock \href {https://doi.org/10.7155/jgaa.00593}
  {\path{doi:10.7155/jgaa.00593}}.

\bibitem{BellenbaumD02}
Patrick Bellenbaum and Reinhard Diestel.
\newblock Two short proofs concerning tree-decompositions.
\newblock {\em Comb. Probab. Comput.}, 11(6):541--547, 2002.
\newblock \href {https://doi.org/10.1017/S0963548302005369}
  {\path{doi:10.1017/S0963548302005369}}.

\bibitem{Bodlaender96}
Hans~L. Bodlaender.
\newblock A linear-time algorithm for finding tree-decompositions of small
  treewidth.
\newblock {\em {SIAM} J. Comput.}, 25(6):1305--1317, 1996.
\newblock \href {https://doi.org/10.1137/S0097539793251219}
  {\path{doi:10.1137/S0097539793251219}}.

\bibitem{Bodlaender98}
Hans~L. Bodlaender.
\newblock A partial \emph{k}-arboretum of graphs with bounded treewidth.
\newblock {\em Theor. Comput. Sci.}, 209(1-2):1--45, 1998.
\newblock \href {https://doi.org/10.1016/S0304-3975(97)00228-4}
  {\path{doi:10.1016/S0304-3975(97)00228-4}}.

\bibitem{BodlaenderDDFLP16}
Hans~L. Bodlaender, P{\aa}l~Gr{\o}n{\aa}s Drange, Markus~S. Dregi, Fedor~V.
  Fomin, Daniel Lokshtanov, and Michal Pilipczuk.
\newblock A {$c^k n$} 5-approximation algorithm for treewidth.
\newblock {\em {SIAM} J. Comput.}, 45(2):317--378, 2016.
\newblock \href {https://doi.org/10.1137/130947374}
  {\path{doi:10.1137/130947374}}.

\bibitem{BodlaenderK96}
Hans~L. Bodlaender and Ton Kloks.
\newblock Efficient and constructive algorithms for the pathwidth and treewidth
  of graphs.
\newblock {\em J. Algorithms}, 21(2):358--402, 1996.
\newblock \href {https://doi.org/10.1006/jagm.1996.0049}
  {\path{doi:10.1006/jagm.1996.0049}}.

\bibitem{BodlaenderK10}
Hans~L. Bodlaender and Arie M. C.~A. Koster.
\newblock Treewidth computations i. upper bounds.
\newblock {\em Inf. Comput.}, 208(3):259--275, 2010.
\newblock \href {https://doi.org/10.1016/j.ic.2009.03.008}
  {\path{doi:10.1016/j.ic.2009.03.008}}.

\bibitem{BodlaenderK11}
Hans~L. Bodlaender and Arie M. C.~A. Koster.
\newblock Treewidth computations {II.} lower bounds.
\newblock {\em Inf. Comput.}, 209(7):1103--1119, 2011.
\newblock \href {https://doi.org/10.1016/j.ic.2011.04.003}
  {\path{doi:10.1016/j.ic.2011.04.003}}.

\bibitem{DBLP:journals/algorithmica/BulianD16}
Jannis Bulian and Anuj Dawar.
\newblock Graph isomorphism parameterized by elimination distance to bounded
  degree.
\newblock {\em Algorithmica}, 75(2):363--382, 2016.
\newblock \href {https://doi.org/10.1007/s00453-015-0045-3}
  {\path{doi:10.1007/s00453-015-0045-3}}.

\bibitem{DBLP:journals/algorithmica/BulianD17}
Jannis Bulian and Anuj Dawar.
\newblock Fixed-parameter tractable distances to sparse graph classes.
\newblock {\em Algorithmica}, 79(1):139--158, 2017.
\newblock \href {https://doi.org/10.1007/s00453-016-0235-7}
  {\path{doi:10.1007/s00453-016-0235-7}}.

\bibitem{Cai96}
Leizhen Cai.
\newblock Fixed-parameter tractability of graph modification problems for
  hereditary properties.
\newblock {\em Inf. Process. Lett.}, 58(4):171--176, 1996.
\newblock \href {https://doi.org/10.1016/0020-0190(96)00050-6}
  {\path{doi:10.1016/0020-0190(96)00050-6}}.

\bibitem{CaoK16}
Yixin Cao.
\newblock Linear recognition of almost interval graphs.
\newblock In Robert Krauthgamer, editor, {\em Proceedings of the Twenty-Seventh
  Annual {ACM-SIAM} Symposium on Discrete Algorithms, {SODA} 2016, Arlington,
  VA, USA, January 10-12, 2016}, pages 1096--1115. {SIAM}, 2016.
\newblock \href {https://doi.org/10.1137/1.9781611974331.ch77}
  {\path{doi:10.1137/1.9781611974331.ch77}}.

\bibitem{CaoM16}
Yixin Cao and D{\'{a}}niel Marx.
\newblock Chordal editing is fixed-parameter tractable.
\newblock {\em Algorithmica}, 75(1):118--137, 2016.
\newblock \href {https://doi.org/10.1007/s00453-015-0014-x}
  {\path{doi:10.1007/s00453-015-0014-x}}.

\bibitem{ChenKX10}
Jianer Chen, Iyad~A. Kanj, and Ge~Xia.
\newblock Improved upper bounds for vertex cover.
\newblock {\em Theor. Comput. Sci.}, 411(40-42):3736--3756, 2010.
\newblock \href {https://doi.org/10.1016/j.tcs.2010.06.026}
  {\path{doi:10.1016/j.tcs.2010.06.026}}.

\bibitem{CourcelleE12}
Bruno Courcelle and Joost Engelfriet.
\newblock {\em Graph Structure and Monadic Second-Order Logic - {A}
  Language-Theoretic Approach}, volume 138 of {\em Encyclopedia of mathematics
  and its applications}.
\newblock Cambridge University Press, 2012.
\newblock \href {https://doi.org/10.1017/CBO9780511977619}
  {\path{doi:10.1017/CBO9780511977619}}.

\bibitem{CyganFKLMPPS15}
Marek Cygan, Fedor~V. Fomin, Lukasz Kowalik, Daniel Lokshtanov, D{\'{a}}niel
  Marx, Marcin Pilipczuk, Michal Pilipczuk, and Saket Saurabh.
\newblock {\em Parameterized Algorithms}.
\newblock Springer, 2015.
\newblock \href {https://doi.org/10.1007/978-3-319-21275-3}
  {\path{doi:10.1007/978-3-319-21275-3}}.

\bibitem{CzerwinskiNP19}
Wojciech Czerwinski, Wojciech Nadara, and Marcin Pilipczuk.
\newblock {Improved Bounds for the Excluded-Minor Approximation of Treedepth}.
\newblock In Michael~A. Bender, Ola Svensson, and Grzegorz Herman, editors,
  {\em 27th Annual European Symposium on Algorithms (ESA 2019)}, volume 144 of
  {\em Leibniz International Proceedings in Informatics (LIPIcs)}, pages
  34:1--34:13, Dagstuhl, Germany, 2019. Schloss Dagstuhl--Leibniz-Zentrum fuer
  Informatik.
\newblock \href {https://doi.org/10.4230/LIPIcs.ESA.2019.34}
  {\path{doi:10.4230/LIPIcs.ESA.2019.34}}.

\bibitem{DellHJKKR16}
Holger Dell, Thore Husfeldt, Bart M.~P. Jansen, Petteri Kaski, Christian
  Komusiewicz, and Frances~A. Rosamond.
\newblock The first parameterized algorithms and computational experiments
  challenge.
\newblock In Jiong Guo and Danny Hermelin, editors, {\em 11th International
  Symposium on Parameterized and Exact Computation, {IPEC} 2016, August 24-26,
  2016, Aarhus, Denmark}, volume~63 of {\em LIPIcs}, pages 30:1--30:9. Schloss
  Dagstuhl - Leibniz-Zentrum f{\"{u}}r Informatik, 2016.
\newblock \href {https://doi.org/10.4230/LIPIcs.IPEC.2016.30}
  {\path{doi:10.4230/LIPIcs.IPEC.2016.30}}.

\bibitem{DellKTW18}
Holger Dell, Christian Komusiewicz, Nimrod Talmon, and Mathias Weller.
\newblock {The {PACE} 2017 Parameterized Algorithms and Computational
  Experiments Challenge: The Second Iteration}.
\newblock In Daniel Lokshtanov and Naomi Nishimura, editors, {\em 12th
  International Symposium on Parameterized and Exact Computation (IPEC 2017)},
  volume~89 of {\em Leibniz International Proceedings in Informatics (LIPIcs)},
  pages 30:1--30:12, Dagstuhl, Germany, 2018. Schloss Dagstuhl--Leibniz-Zentrum
  fuer Informatik.
\newblock \href {https://doi.org/10.4230/LIPIcs.IPEC.2017.30}
  {\path{doi:10.4230/LIPIcs.IPEC.2017.30}}.

\bibitem{Diestel12}
Reinhard Diestel.
\newblock {\em Graph Theory, 4th Edition}, volume 173 of {\em Graduate texts in
  mathematics}.
\newblock Springer, 2012.

\bibitem{DowneyF13}
Rodney~G. Downey and Michael~R. Fellows.
\newblock {\em Fundamentals of Parameterized Complexity}.
\newblock Texts in Computer Science. Springer, 2013.
\newblock \href {https://doi.org/10.1007/978-1-4471-5559-1}
  {\path{doi:10.1007/978-1-4471-5559-1}}.

\bibitem{DBLP:conf/mfcs/EibenGHK19}
Eduard Eiben, Robert Ganian, Thekla Hamm, and O{-}joung Kwon.
\newblock Measuring what matters: {A} hybrid approach to dynamic programming
  with treewidth.
\newblock In Peter Rossmanith, Pinar Heggernes, and Joost{-}Pieter Katoen,
  editors, {\em 44th International Symposium on Mathematical Foundations of
  Computer Science, {MFCS} 2019, August 26-30, 2019, Aachen, Germany}, volume
  138 of {\em LIPIcs}, pages 42:1--42:15. Schloss Dagstuhl - Leibniz-Zentrum
  f{\"{u}}r Informatik, 2019.
\newblock \href {https://doi.org/10.4230/LIPIcs.MFCS.2019.42}
  {\path{doi:10.4230/LIPIcs.MFCS.2019.42}}.

\bibitem{FellowsFLRSV12}
Michael~R. Fellows, Fedor~V. Fomin, Daniel Lokshtanov, Frances~A. Rosamond,
  Saket Saurabh, and Yngve Villanger.
\newblock Local search: Is brute-force avoidable?
\newblock {\em J. Comput. Syst. Sci.}, 78(3):707--719, 2012.
\newblock \href {https://doi.org/10.1016/j.jcss.2011.10.003}
  {\path{doi:10.1016/j.jcss.2011.10.003}}.

\bibitem{DBLP:series/txtcs/FlumG06}
J{\"{o}}rg Flum and Martin Grohe.
\newblock {\em Parameterized Complexity Theory}.
\newblock Texts in Theoretical Computer Science. An {EATCS} Series. Springer,
  2006.
\newblock \href {https://doi.org/10.1007/3-540-29953-X}
  {\path{doi:10.1007/3-540-29953-X}}.

\bibitem{FominGLS09}
Fedor~V. Fomin, Petr~A. Golovach, Daniel Lokshtanov, and Saket Saurabh.
\newblock Clique-width: on the price of generality.
\newblock In Claire Mathieu, editor, {\em Proceedings of the Twentieth Annual
  {ACM-SIAM} Symposium on Discrete Algorithms, {SODA} 2009, New York, NY, USA,
  January 4-6, 2009}, pages 825--834. {SIAM}, 2009.
\newblock URL: \url{http://dl.acm.org/citation.cfm?id=1496770.1496860}.

\bibitem{FominGLS10}
Fedor~V. Fomin, Petr~A. Golovach, Daniel Lokshtanov, and Saket Saurabh.
\newblock Intractability of clique-width parameterizations.
\newblock {\em {SIAM} J. Comput.}, 39(5):1941--1956, 2010.
\newblock \href {https://doi.org/10.1137/080742270}
  {\path{doi:10.1137/080742270}}.

\bibitem{FominGLS14}
Fedor~V. Fomin, Petr~A. Golovach, Daniel Lokshtanov, and Saket Saurabh.
\newblock Almost optimal lower bounds for problems parameterized by
  clique-width.
\newblock {\em {SIAM} J. Comput.}, 43(5):1541--1563, 2014.
\newblock \href {https://doi.org/10.1137/130910932}
  {\path{doi:10.1137/130910932}}.

\bibitem{GuoHNS13}
Jiong Guo, Sepp Hartung, Rolf Niedermeier, and Ondrej Such{\'{y}}.
\newblock The parameterized complexity of local search for {TSP}, more refined.
\newblock {\em Algorithmica}, 67(1):89--110, 2013.
\newblock \href {https://doi.org/10.1007/s00453-012-9685-8}
  {\path{doi:10.1007/s00453-012-9685-8}}.

\bibitem{GuoHK14}
Jiong Guo, Danny Hermelin, and Christian Komusiewicz.
\newblock Local search for string problems: Brute-force is essentially optimal.
\newblock {\em Theor. Comput. Sci.}, 525:30--41, 2014.
\newblock \href {https://doi.org/10.1016/j.tcs.2013.05.006}
  {\path{doi:10.1016/j.tcs.2013.05.006}}.

\bibitem{JansenK21}
Bart M.~P. Jansen and Jari J.~H. de~Kroon.
\newblock {FPT} algorithms to compute the elimination distance to bipartite
  graphs and more.
\newblock In Lukasz Kowalik, Michal Pilipczuk, and Pawel Rzazewski, editors,
  {\em Graph-Theoretic Concepts in Computer Science - 47th International
  Workshop, {WG} 2021, Warsaw, Poland, Revised Selected Papers}, volume 12911
  of {\em Lecture Notes in Computer Science}, pages 80--93. Springer, 2021.
\newblock \href {https://doi.org/10.1007/978-3-030-86838-3_6}
  {\path{doi:10.1007/978-3-030-86838-3_6}}.

\bibitem{JansenKW21}
Bart M.~P. Jansen, Jari J.~H. de~Kroon, and Micha\l{} W\l{}odarczyk.
\newblock Vertex deletion parameterized by elimination distance and even less.
\newblock In {\em Proceedings of the 53rd Annual ACM SIGACT Symposium on Theory
  of Computing}, STOC 2021, page 1757–1769, New York, NY, USA, 2021.
  Association for Computing Machinery.
\newblock \href {https://doi.org/10.1145/3406325.3451068}
  {\path{doi:10.1145/3406325.3451068}}.

\bibitem{JansenKW21-arxiv}
Bart M.~P. Jansen, Jari J.~H. de~Kroon, and Michał Włodarczyk.
\newblock Vertex deletion parameterized by elimination distance and even less.
\newblock {\em CoRR}, abs/2105.04660, 2021.
\newblock URL: \url{https://arxiv.org/abs/2105.04660}.

\bibitem{DBLP:conf/soda/JansenLS14}
Bart M.~P. Jansen, Daniel Lokshtanov, and Saket Saurabh.
\newblock A near-optimal planarization algorithm.
\newblock In Chandra Chekuri, editor, {\em Proceedings of the Twenty-Fifth
  Annual {ACM-SIAM} Symposium on Discrete Algorithms, {SODA} 2014, Portland,
  Oregon, USA, January 5-7, 2014}, pages 1802--1811. {SIAM}, 2014.
\newblock \href {https://doi.org/10.1137/1.9781611973402.130}
  {\path{doi:10.1137/1.9781611973402.130}}.

\bibitem{Kawarabayashi09}
Ken{-}ichi Kawarabayashi.
\newblock Planarity allowing few error vertices in linear time.
\newblock In {\em 50th Annual {IEEE} Symposium on Foundations of Computer
  Science, {FOCS} 2009, October 25-27, 2009, Atlanta, Georgia, {USA}}, pages
  639--648. {IEEE} Computer Society, 2009.
\newblock \href {https://doi.org/10.1109/FOCS.2009.45}
  {\path{doi:10.1109/FOCS.2009.45}}.

\bibitem{Korhonen21}
Tuukka Korhonen.
\newblock A single-exponential time 2-approximation algorithm for treewidth.
\newblock In {\em 62nd {IEEE} Annual Symposium on Foundations of Computer
  Science, {FOCS} 2021, Denver, CO, USA, February 7-10, 2022}, pages 184--192.
  {IEEE}, 2021.
\newblock \href {https://doi.org/10.1109/FOCS52979.2021.00026}
  {\path{doi:10.1109/FOCS52979.2021.00026}}.

\bibitem{KorhonenL22}
Tuukka Korhonen and Daniel Lokshtanov.
\newblock An improved parameterized algorithm for treewidth.
\newblock {\em CoRR}, abs/2211.07154, 2022.
\newblock \href {http://arxiv.org/abs/2211.07154} {\path{arXiv:2211.07154}}.

\bibitem{LokshtanovMS18}
Daniel Lokshtanov, D{\'{a}}niel Marx, and Saket Saurabh.
\newblock Known algorithms on graphs of bounded treewidth are probably optimal.
\newblock {\em {ACM} Trans. Algorithms}, 14(2):13:1--13:30, 2018.
\newblock \href {https://doi.org/10.1145/3170442} {\path{doi:10.1145/3170442}}.

\bibitem{Marx20}
D{\'{a}}niel Marx.
\newblock Four shorts stories on surprising algorithmic uses of treewidth.
\newblock In Fedor~V. Fomin, Stefan Kratsch, and Erik~Jan van Leeuwen, editors,
  {\em Treewidth, Kernels, and Algorithms - Essays Dedicated to Hans L.
  Bodlaender on the Occasion of His 60th Birthday}, volume 12160 of {\em
  Lecture Notes in Computer Science}, pages 129--144. Springer, 2020.
\newblock \href {https://doi.org/10.1007/978-3-030-42071-0\_10}
  {\path{doi:10.1007/978-3-030-42071-0\_10}}.

\bibitem{MarxS12}
D{\'{a}}niel Marx and Ildik{\'{o}} Schlotter.
\newblock Obtaining a planar graph by vertex deletion.
\newblock {\em Algorithmica}, 62(3-4):807--822, 2012.
\newblock \href {https://doi.org/10.1007/s00453-010-9484-z}
  {\path{doi:10.1007/s00453-010-9484-z}}.

\bibitem{NesetrilM12}
Jaroslav Nesetril and Patrice~Ossona de~Mendez.
\newblock {\em Sparsity - Graphs, Structures, and Algorithms}, volume~28 of
  {\em Algorithms and combinatorics}.
\newblock Springer, 2012.
\newblock \href {https://doi.org/10.1007/978-3-642-27875-4}
  {\path{doi:10.1007/978-3-642-27875-4}}.

\bibitem{OumS06}
Sang{-}il Oum and Paul~D. Seymour.
\newblock Approximating clique-width and branch-width.
\newblock {\em J. Comb. Theory, Ser. {B}}, 96(4):514--528, 2006.
\newblock \href {https://doi.org/10.1016/j.jctb.2005.10.006}
  {\path{doi:10.1016/j.jctb.2005.10.006}}.

\bibitem{Pilipczuk17}
Marcin Pilipczuk.
\newblock A tight lower bound for vertex planarization on graphs of bounded
  treewidth.
\newblock {\em Discret. Appl. Math.}, 231:211--216, 2017.
\newblock \href {https://doi.org/10.1016/j.dam.2016.05.019}
  {\path{doi:10.1016/j.dam.2016.05.019}}.

\bibitem{Reed92}
Bruce~A. Reed.
\newblock Finding approximate separators and computing tree width quickly.
\newblock In S.~Rao Kosaraju, Mike Fellows, Avi Wigderson, and John~A. Ellis,
  editors, {\em Proceedings of the 24th Annual {ACM} Symposium on Theory of
  Computing, May 4-6, 1992, Victoria, British Columbia, Canada}, pages
  221--228. {ACM}, 1992.
\newblock \href {https://doi.org/10.1145/129712.129734}
  {\path{doi:10.1145/129712.129734}}.

\bibitem{ReedSV04}
Bruce~A. Reed, Kaleigh Smith, and Adrian Vetta.
\newblock Finding odd cycle transversals.
\newblock {\em Oper. Res. Lett.}, 32(4):299--301, 2004.
\newblock \href {https://doi.org/10.1016/j.orl.2003.10.009}
  {\path{doi:10.1016/j.orl.2003.10.009}}.

\bibitem{RobertsonS95}
N.~Robertson and P.D. Seymour.
\newblock Graph minors {XIII}. {T}he disjoint paths problem.
\newblock {\em Journal of Combinatorial Theory, Series B}, 63(1):65--110, 1995.
\newblock \href {https://doi.org/https://doi.org/10.1006/jctb.1995.1006}
  {\path{doi:https://doi.org/10.1006/jctb.1995.1006}}.

\bibitem{RobertsonS86}
Neil Robertson and Paul~D. Seymour.
\newblock Graph minors. {II.} algorithmic aspects of tree-width.
\newblock {\em J. Algorithms}, 7(3):309--322, 1986.
\newblock \href {https://doi.org/10.1016/0196-6774(86)90023-4}
  {\path{doi:10.1016/0196-6774(86)90023-4}}.

\bibitem{RobertsonS90a}
Neil Robertson and Paul~D. Seymour.
\newblock Graph minors. {IV.} {T}ree-width and well-quasi-ordering.
\newblock {\em J. Comb. Theory, Ser. {B}}, 48(2):227--254, 1990.
\newblock \href {https://doi.org/10.1016/0095-8956(90)90120-O}
  {\path{doi:10.1016/0095-8956(90)90120-O}}.

\bibitem{SaitohYB21}
Toshiki Saitoh, Ryo Yoshinaka, and Hans~L. Bodlaender.
\newblock Fixed-treewidth-efficient algorithms for edge-deletion to interval
  graph classes.
\newblock In Ryuhei Uehara, Seok{-}Hee Hong, and Subhas~C. Nandy, editors, {\em
  {WALCOM:} Algorithms and Computation - 15th International Conference and
  Workshops, {WALCOM} 2021, Yangon, Myanmar, February 28 - March 2, 2021,
  Proceedings}, volume 12635 of {\em Lecture Notes in Computer Science}, pages
  142--153. Springer, 2021.
\newblock \href {https://doi.org/10.1007/978-3-030-68211-8\_12}
  {\path{doi:10.1007/978-3-030-68211-8\_12}}.

\bibitem{Schrijver03}
A.~Schrijver.
\newblock {\em Combinatorial Optimization - Polyhedra and Efficiency}.
\newblock Springer, 2003.

\end{thebibliography}

\end{document}